\documentclass[12pt]{article}
\usepackage[T1]{fontenc}
\usepackage[dvips]{graphicx}
\graphicspath{{images/}}
\usepackage{verbatim}
\usepackage{amsmath,amsthm}
\usepackage{xspace}
\usepackage{pifont}
\usepackage{graphicx}
\usepackage{amssymb}
\usepackage{epic, eepic}
\usepackage{dsfont}
\usepackage{amssymb}
\usepackage{makeidx}
\usepackage{mathrsfs}
\usepackage{exscale}
\usepackage{color} 
\usepackage{overpic} 
\usepackage{bm}
\usepackage{bbm}
\usepackage{booktabs} 
\usepackage{color, colortbl}
\usepackage{subcaption}
\usepackage[numbers]{natbib}

\RequirePackage[colorlinks,citecolor=blue,urlcolor=blue]{hyperref}

\definecolor{Gray}{gray}{0.9}

\usepackage{amsmath,afterpage}
\usepackage{epsf}
\usepackage{graphics,color}
\RequirePackage[colorlinks,citecolor=blue,urlcolor=blue]{hyperref}
\usepackage[export]{adjustbox}

\def\0{\mathbf{0}}

\def\eps{\varepsilon}

\def\lam{\lambda}
\def\rr{\rightarrow}

\def \< {\langle}
\def \> {\rangle}

\def\beqa{\begin{eqnarray}}
\def\eeqa{\end{eqnarray}}
\def\beqas{\begin{eqnarray*}}
\def\eeqas{\end{eqnarray*}}

\def\blue#1{\textcolor{black}{#1}}

\newtheorem{theorem}{Theorem}[section]
\newtheorem{lemma}[theorem]{Lemma}

\newtheorem{proposition}[theorem]{Proposition}

\newtheorem{remark}[theorem]{Remark}
\newtheorem{example}[theorem]{Example}

\newtheorem{definition}[theorem]{Definition}

\numberwithin{equation}{section}
\newcommand{\hatd}[1]{{}}

\newcommand{\bd}{\begin{displaymath}}
\newcommand{\ed}{\end{displaymath}}
\newcommand{\be}{\begin{equation}}
\newcommand{\ee}{\end{equation}}
\newcommand{\bq}{\begin{eqnarray}}
\newcommand{\eq}{\end{eqnarray}}
\newcommand{\bn}{\begin{eqnarray*}}
\newcommand{\en}{\end{eqnarray*}}

\newcommand{\re}{\mathds{R}}

\def\wt{\widetilde}

\def\P{\mathbb{P}}
\def\E{{\mathbb{E}}}
\newcommand{\R}{{\mathbb R}}
\newcommand{\id}{{\rm id}}

\newcommand{\Mid}{{\ \Big|\ }}

\usepackage[normalem]{ulem}

\usepackage{authblk}

\usepackage{mathtools}
\mathtoolsset{showonlyrefs}
\definecolor{blue0}{RGB}{0,77,153} 
\definecolor{red0}{RGB}{179,0,77} 
\definecolor{green0}{RGB}{134,219,76} 
\definecolor{gray0}{RGB}{84,97,110}

\title{Optimal Liquidation with Signals: the General Propagator Case}
\author[1]{Eduardo Abi Jaber\thanks{The first author is grateful for the financial support from the Chaires FiME-FDD, Financial Risks, Deep Finance \& Statistics and Machine Learning and systematic methods in finance at Ecole Polytechnique.}}
\author[2]{Eyal Neuman }
\affil[1]{Ecole Polytechnique, CMAP}
\affil[2]{Department of Mathematics, Imperial College London}

\begin{document}

 \vspace{-0.5cm}
\maketitle

\begin{abstract}
We consider a class of optimal liquidation problems where the agent's transactions create transient price impact driven by a Volterra-type propagator along with temporary price impact. We formulate these problems as maximization of a revenue-risk functionals, where the agent also exploits available information on a progressively measurable price predicting signal. By using an infinite dimensional stochastic control approach, we characterize the value function in terms of a solution to a  free-boundary $L^2$-valued backward stochastic differential equation and an operator-valued Riccati equation. We then derive analytic solutions to these equations which yields an explicit expression for the optimal trading strategy.
We show that our formulas can be implemented in a straightforward and efficient way for a large class of price impact kernels with possible singularities such as the power-law kernel. 
\end{abstract} 


\begin{description}
\item[Mathematics Subject Classification (2010):] 93E20, 60H30, 91G80
\item[JEL Classification:] C02, C61, G11
\item[Keywords:] optimal portfolio liquidation, price impact, propagator models, predictive signals, Volterra stochastic control
\end{description}

\bigskip

\section{Introduction}

Price impact refers to the empirical fact that the execution of a large order
affects the risky asset's price in an adverse and persistent manner  leading to less favourable prices. Propagator models are a central tool in describing this phenomena mathematically. 
This class of models provides deep insight into the nature of price impact and price dynamics.  It expresses price moves in terms of the influence of past trades, which gives reliable reduced form view on the limit order book. It provides interesting insights on liquidity, price formation and on the interaction between different market participants through price impact. The model's tractability provides a convenient formulation for stochastic control problems arising from optimal execution \citep{bouchaud_bonart_donier_gould_2018, gatheral2010no}. 
More precisely, if the trader's holdings in a risky asset is denoted by $Q=\{Q_t\}_{t\geq 0}$, then the asset price $S_t$ is given by, 
$$
S_t = S_0+ \int_0^tG(t-s)dQ_s + M_t,
$$
where $M$ is a martingale and the price impact kernel $G$ is called a propagator. It can be shown both from theoretical arguments such as market efficiency paradox and empirically that $G(t)$ must decay for large values of $t$, therefore the integral on the right-hand-side of the above equation is referred to as transient price impact (see e.g. \citet[Chapter 13]{bouchaud_bonart_donier_gould_2018}). The two extreme cases where $G$ is Dirac's delta and when $G=1$ are referred to as temporary price impact and permanent price impact, respectively. They are core features in the well known Almgren-Chriss model \citep{AlmgrenChriss1,OPTEXECAC00}, up to a multiplicative constant. 

Considering the adverse effect of the price impact on the execution price, a trader who wishes to minimize her
trading costs has to split her order into a
sequence of smaller orders which are executed over a finite time
horizon. At the same time, the trader also has an incentive to execute
these split orders rapidly because she does not want to carry the risk
of an adverse price move far away from her initial decision
price. This trade-off between price impact and market risk is usually
translated into a stochastic optimal control problem where the trader
aims to minimize a risk-cost functional over a suitable class of
execution strategies, see~\citep{cartea15book, GatheralSchiedSurvey, GokayRochSoner, Gueant:16, citeulike:12047995,N-Sch16} among others. In practice however, apart from focusing on the trade-off between price impact and market risk, many traders and trading algorithms also strive for using short term price predictors in their dynamic order execution schedules. Most of such documented predictors
relate to order book dynamics as discussed in~\cite{leh16moun,Lehalle-Neum18, citeulike:12820703, cont14}.  From the modelling point of view, incorporating signals into execution problems translates into taking into consideration a non-martingale price process, which changes the problem significantly. The resulting optimal strategies in this setting are often random and in particular signal-adaptive, in contrast to deterministic strategies, which are typically obtained in the martingale price case \citep{brigo18, Brigo-21}. Results on optimal trading with signals but without a transient price impact component (i.e. $G=0$) were derived in~\citep{Car-Jiam-2016,Lehalle-Neum18,BMO:19}.  

The special case where the propagator is exponential simplifies the liquidation problem, as the transient price impact can be written as a state variable and the problem becomes Markovian. The exponential propagator case was first solved by \citet{Ob-Wan2005} and by \citet{LorenzSchied}, where further extensions were derived by \citep{GraeweHorst:17,ChenHorstTran:19,SchiedStrehleZhang} among others. In this class of problems, sometimes temporary impact is also included, but trading signals are not taken into account, which leads to deterministic optimal strategies.  In~\citet{NeumanVoss:20}, the liquidation problem with an exponential propagator and a general semimartingale signal was solved and an explicit signal adaptive optimal strategy was derived. 

Results on optimal liquidation problems with a general class of price impact kernels are scarce as the associated stochastic control problem is non-Markovian and often singular. Indeed the transient price impact term and hence the asset execution price encode the entire trajectory of the agent's trading. A first contribution towards solving this problem was made by \citet{GSS}, who solved the deterministic case without signals and without a risk-aversion term. They minimised the following energy functional over left-continuous and adapted strategies  $Q=\{Q_t\}_{t\geq 0}$ with a fuel constraint, i.e. $Q_{T+}=0$ 
$$
C(Q) = \int_{[0,T]}\int_{[0,T]} G(|t-s|)dQ_sdQ_t. 
$$
Here $C(Q)$ represents the trader's transaction costs and $Q$ as before, is the trader's holdings in the risky asset. Under the assumption that the convolution kernel $G$ is non-constant, nonincreasing, convex and integrable, a necessary and sufficient first order condition in the form of a Fredholm equation was derived in \cite{GSS}. This condition was used in order to derive the optimal strategy for several examples of kernels  including the power law kernel. These results were further improved by \citet{Alf-Schied-13} who assumed that $G$ is completely monotone and satisfies $G''(+0)<\infty$, which excludes the case of the fractional kernel. They characterised the optimal strategy in terms of an infinite dimensional Riccati equation.
 
The main objective of this paper is to solve a general class of liquidation problems in the presence of linear transient price impact, which is induced by a nonnegative-definite Volterra-type propagator, along with taking into account a progressively measurable signal. We formulate these problems as a maximization of revenue-risk functionals over a class of progressively measurable strategies. Our solution to these problems solves an open problem put forward in~\cite{Lehalle-Neum18} and also significantly extends the deterministic theory of \citet{Alf-Schied-13}. We develop a novel approach to tackle these problems by using tools from  stochastic Volterra control theory. Our methodology complements and extends the growing literature on linear-quadratic stochastic Volterra  problems \cite{abi2021linear, wang2018linear, kleptsyna2003linear, duncan2013linear, abi2020markowitz,hamaguchi2022linear, aj-22}, and allows for novel explicit formulas even in the case of non-convolution kernels. Indeed, our derivation characterizes the value function in terms of a quadratic dependence on an operator-valued Riccati equation and linear dependence in a solution to a non-standard free-boundary $L^2$-valued backward stochastic differential equation. We then derive analytic expressions for the solutions of these equations which in turn yields an explicit expression for the optimal trading strategy
(see Theorem \ref{eq:MainTheorem} and Proposition \ref{prop-exp-u}). Finally, we show that our formulas can be implemented in a straightforward and efficient way for a large class of price impact kernels.\footnote{We also provide the code of our implementation  at \url{https://colab.research.google.com/drive/1VQasI92YhdBC0wnn_LxMkkx_45VyK1yQ}.}  In particular, our results cover the case of non-convolution singular price impact kernels such as the power-law kernel (see Remark \ref{rem-examples} for additional examples). 

The results in this paper significantly improve the results of~\cite{NeumanVoss:20} as we allow for a general Volterra propagator instead of an exponential one. This turns the stochastic control problem to become non-Markovian as the state variables (e.g. the execution price) depend on the entire trading trajectory, unlike the exponential kernel case where the transient price impact could be regraded as a mean-reverting state variable hence the problem become Markovian (see Lemma 5.3 in \cite{NeumanVoss:20}). We also generalise the price process dynamics in \cite{NeumanVoss:20}, which was assumed to be a semimartingale, while here we assume that it is a progressively measurable process. 

Our main results also substantially generalise the results of \citet{Alf-Schied-13} in various directions. First, in contrast to \cite{Alf-Schied-13}, we assume that the price process is non-martingale, which turns the problem from deterministic optimisation to stochastic, and introduces new ingredients in the value function, {which depend on $L^2$-valued free-boundary BSDE  (see \eqref{eq:chi}, \eqref{eq:candidatevaluefun} and \eqref{back-theta}). }Moreover it is assumed in \cite{Alf-Schied-13} that $G$ is a convolution kernel which is completely monotone and satisfies $G''(+0)<\infty$. In this work we show that these assumptions are {not necessary} and in fact power law kernels of the form $G(t) =t^{-\beta}$ for $0<\beta<1/2$ are included in our class of admissible kernels. The solution to the problem in \cite{Alf-Schied-13} is given in terms of an infinite dimensional Riccati equation which takes values in $\mathbb{R}$ (see eq. (5) and (6) therein). This could be compared with our operator-valued Riccati equation in \eqref{eq:riccati_psiBold} which is one of the main ingredients of the solution (see \eqref{eq:optimalcontrol_monotone} and \eqref{psi-gamma}). However, as stated in Section 1.3 of  \cite{Alf-Schied-13}, their Riccati equation in general cannot be solved explicitly, and the only tractable example provided is when $G$ is a finite sum of exponential kernels. In this work we solve explicitly the operator Riccati equation (see \eqref{def:riccati_operator}) along with all the other ingredients of the value function. Moreover, in Section \ref{sec-simulation} we give a detailed numerical scheme to implement these explicit solutions as a finite-dimensional projection of the operators. Lastly, in contrast with  \cite{Alf-Schied-13} we incorporate a risk aversion term into the cost functional \eqref{eq:optimalvalue_monotone}, which has an important practical role as it reflects the risk of holding inventory.

Finally, our paper is also related to a recent work by~\citet{forde_22}, where a  specific example of an optimal liquidation problem with
power-law transient price impact, a Gaussian signal, and without a risk-aversion term was studied. In the main result of \cite{forde_22}, a first order condition for the solution was derived in terms of Fredholm integral equations of the first kind. Then, examples for explicit solutions were worked out for a specific choice of signals, which are convolution of fractional kernels with respect to Brownian motion.

 
\paragraph{Organization of the paper.} The remainder of the paper is structured as follows. In Section \ref{sec:setup}, the class of liquidation problems is defined. In Section \ref{sec-trans}, we transform the cost functional and state variables in order to formulate the problem in an infinite dimensional setting. Section \ref{sec-results} is dedicated to the presentation of the main results namely Theorem \ref{eq:MainTheorem} and Proposition \ref{prop-exp-u}. In Section \ref{sec-simulation}, we provide a numerical scheme for plotting the optimal strategy in \ref{sec-simulation} and provide illustrative examples for such computations. Sections \ref{sec-thm-pf} and \ref{sec-prop-expl} are dedicated to the proofs of Theorem \ref{eq:MainTheorem} and Proposition \ref{prop-exp-u}, respectively. Finally, Sections \ref{sec-pf-id}--\ref{sec-lem-inv} contain proofs to some auxiliary results.

\section{Model setup and problem formulation} \label{sec:setup}
We present the class of optimal liquidation problems which are studied in this paper. Let $T>0$ denote a finite deterministic time horizon and fix a filtered probability space $(\Omega, \mathcal F,(\mathcal F_t)_{0 \leq t\leq T}, \P  )$ satisfying the usual conditions of right continuity and completeness. We fix a progressively measurable process~$P=(P_t)_{0 \leq t\leq T}$
 %
satisfying 
\be \label{ass:P} 
 { \mathbb E\left[ \sup_{ 0\leq t \leq r  \leq T} \left( \mathbb E [P_r | \mathcal F_t] \right)^2 \right] <\infty}.
\ee
The technical assumption \eqref{ass:P} ensures in particular that 
$$  \mathbb E\left[\int_0^T P^{2}_{s}ds \right] < \infty,
$$
and is readily satisfied for instance for our chief numerical example considered in Section \ref{s:numsub} below on  $P_t = \int_0^t I_s ds + M_t$ where $M$ is a  martingale and $I$ is an Ornstein-Uhlenbeck process.

We consider a trader with an initial position of {$q \in \mathbb{R}$} shares in a risky asset. The number of shares the trader holds at time $t\in [0,T]$ is prescribed as 
    \begin{align} \label{def:Q}
    Q_t^u := q -\int_0^t u_s ds, 
    \end{align}
where $(u_s)_{s \in [0,T]}$ denotes the trading speed which is chosen from the set of admissible strategies
\be \label{def:admissset} 
\mathcal A :=\left\{ u \, : \, u \textrm{ progressively measurable s.t. }  \mathbb E\Big[  \sup_{t \leq T} u_t^2 \Big] <\infty \right\}.
\ee
We assume that the trader's trading activity causes price impact on the risky asset's execution price. 
We consider a Volterra kernel $G:[0,T]^2 \rr \mathbb{R}_+$, that is $G(t,s)=0$ for $s\geq t$, within a certain class of square-integrable admissible kernels which will be defined in Definition~\ref{def-ker-admis} below. Then, we introduce the actual price $S^u$ in which the orders are executed along a certain admissible strategy $u$:
\be \label{def:S}
S_t^u := P_{t} - \lam u_t -  Z^u_t, \qquad 0 \leq t \leq T, 
\ee
where $P$ plays the role of the unaffected price of the risky asset and
\be \label{z-def} 
Z_t^u:= h_{0}(t)+ \int_0^t G(t,s)u_s ds, \qquad 0 \leq t \leq T, 
\ee
for some {continuous} 
deterministic function $h_{0}:[0,T] \to \mathbb R$.

 Specifically, the trader's transaction not only instantaneously affects the execution price in~\eqref{def:S} in an adverse manner through a linear temporary price impact $\lambda > 0$ \`a la~\citet{OPTEXECAC00}; it also induces a longer lasting price distortion $Z^u$ because of the linear transient price impact (see e.g. \citet{GSS}). 
 
We now suppose that the trader's optimal trading objective is to unwind her initial position $q \in \R$ in the presence of temporary and transient price impact, along with taking into account the asset's general price, through maximizing the performance functional 
\begin{equation} \label{def:objective}
J(u) := \mathbb{E} \Bigg[ \int_0^T (P_t - Z^u_t) u_t dt - \lambda \int_0^T u^2_t dt + Q_T^u P_T -\phi \int_0^T (Q_t^u)^2 dt - \varrho (Q_T^u)^2 \Bigg], 
\end{equation}
via her selling rate $u \in \mathcal A$. The first three terms in~\eqref{def:objective} represent the trader's terminal wealth;  that is, her final cash position including the accrued trading costs which are induced by temporary and transient price impact as prescribed in~\eqref{def:S}, as well as her remaining final risky asset position's book value. The fourth and fifth terms in~\eqref{def:objective} implement a penalty $\phi {\geq} 0$ and $\varrho {\geq} 0$ on her running and terminal inventory, respectively. Also observe that $J(u) < \infty$ for any admissible strategy $u \in \mathcal A$.

The goal of this paper is to find the optimal strategy $u^*$ that maximizes the trader's performance functional: 
\begin{align} \label{control-prob}
 J(u^*)=\sup_{u \in \mathcal A} J(u).
\end{align}
Our main result summarised in Theorem~\ref{eq:MainTheorem} and Proposition~\ref{prop-exp-u} shows that, remarkably,  the problem can be solved explicitely despite the path-dependency of the model. More precisely, we show that the optimal strategy  $u^*$ is explicitly given by the solution to a  linear Volterra  equation of the form
\be 
 u^*_t = a_t + \int_0^t B(t,s)u_s^*ds,
 \ee
where $\{a_t\}_{t\in [0,T]}$ is a stochastic process that depends linearly on the price process $P$ and $B$ is a deterministic kernel. Both $a$ and $B$ are given explicitly in  \eqref{eq:aB} below, in terms of the inputs of the model and  of the price impact kernel $G$, under very mild assumptions on $G$ detailed in the next paragraph. Such expressions lend themselves naturally to numerical discretization schemes as shown in Section~\ref{sec-simulation}.

After specifying the optimization problem \eqref{control-prob} we introduce some additional assumptions on to the class of \emph{price impact kernels} or \emph{propagators}, which will be used throughout this paper. 
We say that a Volterra kernel $G:[0,T]^2 \to \R_{+}$ with $G(t,s)=0$ whenever $s \geq t$, is nonnegative definite if for every $f\in L^2\left([0,T],\mathbb R\right)$ we have  
\be \label{pos-def}
\int_{0}^T\int_{0}^T\big(G(t,s)+ G(s,t) \big)f(s)f(t)dsdt \geq 0. 
\ee

\begin{remark}  \label{rem-pos-k} 
Note that when 
\begin{align}\label{eq:Gconv}
    G(t,s) = \mathds 1_{\{s<t\}}H(t-s),
\end{align} we can replace \eqref{pos-def} with the following condition   
\be \label{pos-def2}
\int_{0}^T\int_{0}^T H(|t-s|)f(s)f(t)dsdt \geq 0. 
\ee
Note that \eqref{pos-def2} is the main assumption on the price impact kernel in \citet{GSS}. As discussed in Section 2 of \cite{GSS}, for the case where the price process $P$ is a martingale (i.e.~there is no price predicting signal), the coefficients $\lambda, \phi =0$ and we restrict to strategies with 
a fuel constraint, that is $Q_T=0$, then \eqref{pos-def2} ensures that the model does not admit price manipulations, and in particular round trips (see Definition 2.5 therein and discussion afterwards). This fact can be extended easily to the case of positive $\lambda, \phi$ as this adds quadratic terms to the cost functional \eqref{def:objective}. However, once the price process is no longer a martingale, as in the setting of this paper, round trips are possible. We refer to figure 3 in \cite{NeumanVoss:20} for some illustrations of this phenomenon when H is an exponentially decaying kernel.
\end{remark}

Volterra convolution kernels of the form \eqref{eq:Gconv} are nonnegative  definite kernels whenever the function $H$ is bounded, non-increasing  and convex  (see Example 2.7 in \cite{GSS}). The following lemma, which is a slight generalization of Bochner's theorem in one direction, gives an additional characterisation for an important subclass of nonnegative definite kernels. The proof of Lemma \ref{lemma-pos-def-ker} is postponed to Section \ref{sec-lem-inv}. 
\begin{lemma} \label{lemma-pos-def-ker} 
Let $G$ be of the form \eqref{eq:Gconv} with $H:(0,\infty) \rr [0,\infty]$. If $H$ can be represented as
\be \label{g-spec} 
H(t)= \int_{\mathbb{R}_+} e^{-xt} \mu(dx),  \quad 0\leq t\leq T, 
\ee
where $\mu$ is a nonnegative measure, then $G$ is nonnegative definite. 
\end{lemma} 

We define the following class of admissible kernels, which will be considered throughout this paper.  
\begin{definition} [Class of admissible kernels $\mathcal G$]  \label{def-ker-admis} 
We say that a nonnegative definite Volterra kernel $G:[0,T]^2 \mapsto \R_{+}$ is in the class of kernels $\mathcal G$ if it satisfies the following conditions:  
\begin{align}\label{eq:assumtionG}
	\begin{split}
	\sup_{t\leq T} \int_0^T |G(t,s)|^2 ds  + \sup_{s\leq T} \int_0^T |G(t,s)|^2 dt< \infty, \\
	\lim_{h\to 0} \int_0^T |G(t + h, s) - G(t, s)|^2 ds = 0, \quad t \leq T.
	\end{split}
\end{align}  
\end{definition} 
\begin{remark} \label{rem-kernels} 
Note that any convolution kernel  $G(t,s)= \mathds{1}_{\{s<t\}} H(t-s),$ with $H \in L^2([0,T], \R)$ satisfies \eqref{eq:assumtionG}.
\end{remark}

\begin{example} \label{rem-examples} 
We present some typical examples for price impact kernels which belong to the class $\mathcal G$. The first three kernels are of convolution type \eqref{eq:Gconv}.
\begin{enumerate}
\item  In \cite{bouchaud2004fluctuations, gatheral2010no} among others the following kernel was introduced:   
$$ H(t)=\frac{\ell_{_{0}}}{(\ell_{0}+t)^{\beta}}, \quad  \textrm{for }  \beta >0,$$
where $\ell_{0}>0$ is a constant. 
 \item Kernels of the form 
 $$ 
 H(t)=\frac{1}{t^{\beta}}, \quad  \textrm{for }  0< \beta < 1/2,
 $$
were proposed by Gatheral in \cite{gatheral2010no}. Thanks to Lemma \ref{lemma-pos-def-ker} and to the spectral representation of fractional kernels (see e.g. eq. (1.3) in \cite{abi-euch}) we observe that this singular kernel is indeed in $\mathcal G$. 
 \item The case where $H (t) =  e^{-\rho t}$, for some constant $\rho >0$, was proposed by Obizhaeva and Wang \cite{Ob-Wan2005}. Clearly any linear combination of such kernel is also applicable. 
 \item The following non-convolution kernel was used in order to model price impact in bonds trading (see Section 3.1 of  \citet{brigo2020}): 
 $$
G(t,s) = f(t-T) H(t-s)  \mathds 1_{\{s<t\}},
$$
where $H$ is a usual decay kernel as in the above examples and $f$ is a bounded function satisfying $f(0)=0$, due to the terminal condition on the bond price.  
 \end{enumerate}
\end{example} 


\section{Transformation of the performance functional} \label{sec-trans} 
Considering the state process $Z^u$ in \eqref{z-def}, we notice that the stochastic control problem \eqref{control-prob} is path-dependent. In this section we transform the performance functional  \eqref{def:objective} and state variables so they could fit an infinite-dimensional stochastic control famework.

One can notice at this stage that \eqref{def:objective} is linear-quadratic in $(Z,Q)$. For convenience, we will incorporate the terminal quadratic term to the running cost by using integration and \eqref{def:Q}: 
\be \label{d-sq} 
 (Q^u_T)^2 = q^2 - 2\int_0^T Q^u_s u_s ds.
\ee
We moreover define  
\be \label{y-z} 
Y_t^u :=  Z_t^u - 2\varrho Q_t^u. 
\ee
From \eqref{def:objective}, \eqref{d-sq} and \eqref{y-z} we get that $J$ defined in \eqref{def:objective} can be re-written as
\begin{equation} \label{def:objective2}
 J(u) =\mathbb{E} \Bigg[ \int_0^T (P_t -  Y^u_t) u_t dt - \lambda \int_0^T u^2_t dt +Q_T^u P_T  -\phi \int_0^T (Q_t^u)^2 dt  \Bigg] -\varrho q^2.
\end{equation}

We further define,  
\be \label{h-tilde} 
\begin{aligned}
\tilde h_0(t) &:= h_0(t) - 2\varrho q \\
  \tilde G(t,s) &:=2 \varrho  \mathds{1}_{\{s<t\}} + G(t,s), \quad t,s \leq T. 
\end{aligned}
\ee
Together with \eqref{def:Q}, \eqref{z-def} and \eqref{y-z} we get, 
\be  \label{y-h-tilde} 
Y_t^u  = \tilde h_0(t) + \int_0^t \tilde G(t,s) u_s ds. 
\ee 
We further introduce a new state variable, 
\be
 X^u:=(Y^u, Q^u)^\top. 
\ee 
Note that from \eqref{def:Q} and \eqref{y-h-tilde} it follows that we can rewrite $X^u$ as follows: 
$$ 
X^u_t = g_0(t)  +  \int_0^t K(t,s) u_s ds, 
$$
where  
\be \label{g0-k-def} 
 g_0(t) := (\tilde h_0(t), q )^\top, \quad K(t,s) := (\tilde G(t,s), - \mathds{1}_{\{s\leq t\}})^\top. 
 \ee
We also define the so-called \emph{controlled adjusted forward process} as follows: 
\be \label{eq:gu}
\begin{aligned}
g^u_t(s) &= \mathds{1}_{\{s\geq t\}} \E\left[X^u_s - \int_t^s K(s,r) u_r dr  \Big | \mathcal F_t \right]  \\
&= \mathds{1}_{\{s\geq t\}} \left(g_0(s) + \int_0^t K(s,r)u_r dr\right) \\ 
&= \mathds{1}_{\{s\geq t\}} \left( \tilde h_{0}(s) + \int_{0}^{t}\tilde G(s,r)u_{r}dr, Q^u_{t}\right)^{\top}.
\end{aligned}
\ee
Note that the second component of $g^u_t(s)$ is always equal to $Q^u_t$ (since $Q$ is Markovian) and that 
\be \label{g-def} 
g^u_t(t)=X^u_t = (Y^{u}_{t}, Q_{t}^{u} )^\top, \quad t \leq T.
\ee

\section{Main Results} \label{sec-results} 
In this section, we derive explicitly the maximiser of \eqref{control-prob}. Before stating this result we introduce some essential definitions of function spaces, integral operators and stochastic processes.   
\subsection{Function spaces, integral operators} 
We denote by $\langle \cdot, \cdot \rangle_{L^2}$ the inner product on $L^2([0,T], \R^2)$, that is 
\be \label{in-prod} 
\langle f, g\rangle_{L^2} := \int_0^T f(s)^{\top} g(s) ds, \quad f,g\in L^2\left([0,T],\mathbb R^2\right). 
\ee
We define $L^2\left([0,T]^2,\mathbb R^{2 \times 2}\right)$ to be the space of measurable kernels $\Sigma:[0,T]^2 \to \R^{2\times 2}$ such that 
\begin{align*}
\int_0^T \int_0^T |\Sigma(t,s)|^2 dt ds < \infty.
\end{align*}
The notation $|\cdot|$ stands for a matrix norm, and in particular we have
$$
\int_0^T \int_0^T |\Sigma_{i,j}(t,s)|^2 dt ds < \infty, \quad \textrm{for all } i,j =1,2. 
$$
For any $\Sigma, \Lambda \in  L^2\left([0,T]^2,\mathbb R^{2\times 2}\right)$ we define the $\star$-product as follows
\begin{align*}
(\Sigma \star \Lambda)(s,u) := \int_0^T \Sigma(s,z) \Lambda(z,u)dz, \quad  (s,u) \in [0,T]^2,
\end{align*}
which is a well-defined kernel in $L^2\left([0,T]^2,\mathbb R^{2\times 2}\right)$ due to Cauchy-Schwarz inequality.  For any  kernel $\Sigma \in L^2\left([0,T]^2,\mathbb R^{2\times 2}\right)$, we denote by {$\boldsymbol \Sigma$} the integral operator   induced by the kernel $\Sigma$ that is 
\begin{align*} 
({\boldsymbol \Sigma} g)(s):=\int_0^T \Sigma(s,u) g(u)du,\quad g \in L^2\left([0,T],\mathbb R^2\right).
\end{align*}
$\boldsymbol \Sigma$ is a linear bounded operator from  $L^2\left([0,T],\mathbb R^2 \right)$ into itself. 
For $\boldsymbol{\Sigma}$ and $\boldsymbol{\Lambda}$ that are two integral operators induced by the kernels $\Sigma$ and $\Lambda$  in $L^2\left([0,T]^2,\mathbb R^{2\times 2}\right)$, we denote by $\boldsymbol{\Sigma}\boldsymbol{\Lambda}$ the integral operator induced by the kernel $\Sigma\star \Lambda$.

We denote by $\Sigma^*$ the adjoint kernel of $\Sigma$ for $\langle \cdot, \cdot \rangle_{L^2}$, that is 
\begin{align*} 
\Sigma^*(s,u) &= \; \Sigma(u,s)^\top, \quad  (s,u) \in [0,T]^2,
\end{align*}
and by $\boldsymbol{\Sigma}^*$ the corresponding adjoint integral operator.  

We recall that an operator $\boldsymbol{\Sigma}$ as above is said to be non-negative definite if $\langle \boldsymbol{\Sigma}f,f\rangle_{L^{2}} \geq0$ for all $f\in L^2\left([0,T],\mathbb R^2\right)$. It is said to be positive definite if $\langle\boldsymbol{\Sigma}f,f\rangle_{L^{2}}> 0$ for all $f\in L^2\left([0,T],\mathbb R^2\right)$ not identically zero.

\subsection{Essential operators for our setting}


\paragraph{The $\Gamma^{-1}_t$ operator:}
Recall that $\tilde G$ was defined in \eqref{h-tilde}. We define $\boldsymbol{\tilde G}_t$ as the operator induced by the kernel $\tilde G(s,u)\mathds 1_{\{u \geq t\}}$. We introduce 
\begin{align}\label{eq:schur}
    \boldsymbol{D}_t := 2\lambda \id +  (\boldsymbol{\tilde G}_t + \boldsymbol{\tilde G}^*_t) + 2\phi  \boldsymbol{1}^*_t \boldsymbol{1}_t,
\end{align}
where $\id$ is the idendity operator, i.e.~$(\id f)(t)=f(t)$, $\boldsymbol{1}_t$ is the integral operator induced by the kernel 
\be \label{op-one} 
\mathds 1_t(u,s) := \mathds 1_{\{u \geq  s\}} \mathds 1_{\{s \geq t\}} .
\ee
The following lemma, which is proved in Section \ref{sec-lem-inv}, provides the invertibility of $\boldsymbol{D}_t$, which is essential for upcoming definitions.  

\begin{lemma} \label{lemmma-inv-op} 
Assume that $\lambda>0$ and $\varrho, \phi \geq 0$. Then, for any $G \in \mathcal G$, the operator $\boldsymbol{D}_t$ is {positive} definite, self-adjoint and invertible. 
\end{lemma}

Using Lemma~\ref{lemmma-inv-op}, we can therefore define an operator $\boldsymbol{\Gamma}_t^{-1}$ by 
 \be \label{op-gam-inv} 
\begin{aligned}
\boldsymbol {\Gamma}_t^{-1} :=  \left(\begin{matrix}
 \boldsymbol{D}_t^{-1}  &-2 \phi \boldsymbol{D}_t^{-1} \boldsymbol{1}^*_t \\
-2\phi \boldsymbol{1}_t \boldsymbol{D}_t^{-1} & - {2\phi} \id + 4\phi^2  \boldsymbol{1}_t \boldsymbol{D}_t^{-1} \boldsymbol{1}^*_t
\end{matrix}\right).
\end{aligned}
\ee 
We note that for $\phi >0$, $\boldsymbol \Gamma^{-1}_t$ is the inverse of the operator 
\be \label{op-gamma} 
\begin{aligned}
\boldsymbol{\Gamma}_t 
&= \left(\begin{matrix}
\boldsymbol{D}_t-2\phi  \boldsymbol{1}^*_t \boldsymbol{1}_t  & -\boldsymbol{1}^*_t \\
-\boldsymbol{1}_t & - \frac 1{2\phi} \id 
\end{matrix}\right). 
\end{aligned}
\ee
Note also that $\boldsymbol \Gamma^{-1}$ solves an operator Riccati equation (see \eqref{psi-gamma}  and Lemma~\ref{L:Psi}).

\subsection{Essential stochastic processes} 
\paragraph{The process $\Theta$:} 
For convenience we introduce the following notation, 
\be \label{ind-def} 
\mathds{1}_t(s):= \mathds{1}_{\{ s \geq t\}},
\ee
{and let $e_1 :=(1,0)^\top$. } 
 We define $\Theta= \{\Theta_{t}(s) : t\in [0,s],\, s\in [0,T]\}$ as follows,  
\be \label{th-sol} 
\Theta_t(s):=- \left(\boldsymbol {\Gamma}_t^{-1}  \mathds{1}_t \E\left[P_{\cdot}- P_T \Mid \mathcal F_t\right] e_1 \right)(s).
\ee
Note that $\Theta$ solves the $L^2$-valued BSDE \eqref{back-theta} (see Proposition \ref{prop-theta}). 

\paragraph{The auxiliary process $\chi$: } 
For $P$ as in \eqref{ass:P} we define the following martingale  
\be \label{m1-mart} 
M_t := \E[P_T \mid \mathcal F_t], \quad 0\leq t\leq T.
\ee
For $K$ as in \eqref{g0-k-def}, we use the notation
\be \label{K-t}
K_t(s):= K(s,t).
\ee
Finally we define the stochastic process $\chi = \{ \chi_t\}_{t \in [0,T]}$ as follows, 
\begin{align}\label{eq:chi}
 \chi_t := -2\varrho q^2 +  \int_t^T \frac 1 {2\lambda} \E\left[   \left(  P_s -  P_T + \langle \Theta_s , K_s \rangle_{L^2} \right)^2  \mid \mathcal F_t \right] ds, \quad 0\leq t \leq T. 
 \end{align}
 Note that $\chi$ in \eqref{eq:chi} solves the following BSDE 
\begin{align}\label{eq:BSDEchi}
\begin{split}
d\chi_t &= \dot{\chi}_t dt + d\widetilde M_t , \quad \chi_T = -2\varrho q^2,  \\
\dot{\chi}_t &= -\frac{1}{2\lambda} \left(P_t - M_t + \langle \Theta_t, K_t\rangle_{L^2}\right)^2,
\end{split}
\end{align}	
where $\widetilde M$ is the following martingale  
\begin{align}\label{eq:zxprop}  
    \widetilde M_t: =  \frac{1}{2\lambda}\mathbb E\left[ \int_0^T  \left(  P_s -  P_T + \langle \Theta_s , K_s \rangle_{L^2} \right)^2 ds  \mid \mathcal F_t \right]. 
\end{align}


\subsection{Solution to the liquidation problem} 
Now we are ready to present our main results. Given the linear-quadratic structure of the performance functional $J$ in \eqref{def:objective2} and  the conditioned state variable $g^{u}$ in 
\eqref{eq:gu}, it is natural to consider a candidate for the value function of the form
\begin{align}\label{eq:candidatevaluefun}
V_t^{u}&:= \frac{1}{2}\left(\langle g_t^{u}, \boldsymbol {\Gamma}_t^{-1} g_t^{u} \rangle_{L^2} + 2\langle \Theta_t, g_t^{u} \rangle_{L^2} +  2\mathbb E[P_T \mid \mathcal F_t] Q_t^{u} +  \chi_t \right), \  0\leq t\leq T,
\end{align}  
which is the infinite dimensional analogue of standard liquidation problems with signals \cite{NeumanVoss:20}. {Indeed the solution presented in Theorem 3.2 of \cite{NeumanVoss:20} for the Markovian case (i.e. for exponential propagator) shows that the optimal trading speed is affine with respect to the state variables in \eqref{g-def} and is also affine with respect to the signal $\mathbb E[P_T \mid \mathcal F_t]$. If one plugs-in this ansatz to the performance functional $J$ in \eqref{def:objective2}, then it would follow that the value function depends in a linear quadratic manner in the state variable $g_t^{u}$ and linearly with respect to the signal. }

In the following definition we define the optimal control and value function in our infinite dimensional setting. Recall the definition of the set of admissible controls $\mathcal A$  in \eqref{def:admissset}. 
\begin{definition} \label{def-sol} 
We say that $u^* \in \mathcal A$ is an optimal strategy and that  $\{V_t^{u^*}\}_{t\geq 0}$ given by \eqref{eq:candidatevaluefun} is the optimal \emph{value process} of the cost functional \eqref{def:objective}  if we have for all 
$0 \leq t \leq T$, 

	\begin{equation}	\label{eq:optimalvalue_monotone}
	V_t^{u^*} = \;  \operatorname*{ess~sup}_{u \in \mathcal A_t(u^*)} \E\left[\int_{t}^{T}  \left( (P_s-Y^u_s) u_s - \lambda u_s^2 - \phi (Q^u_s)^2 \right)ds +  P_T Q^u_T  \Big| \mathcal F_t\right]- \varrho q^2, \quad \P-a.s.
	\end{equation}
where 
\bd \label{admis-u} 
\mathcal A_t(u)=\{ u' \in \mathcal A: u_s' = u_s, \; \text{on } [0,t]\times \Omega, \;  dt\otimes d\P-a.e. \}.
\ed
\end{definition} 

\begin{remark} 
Note that for $u^*$ as in Definition \ref{def-sol} we specifically have for $t=0$,  
$$
V_0^{u^*} = \sup_{u \in \mathcal A}  J(u).
$$
\end{remark} 
Now we are ready to present our main result. We fix a square-integrable deterministic function $h_{0}:[0,T] \rr \re$  as in \eqref{z-def} and $G$ from the class of price impact kernels $\mathcal G$ from Definition \ref{def-ker-admis}. 
We also recall that $Y^{u}$ and $g^{u}$ were defined in \eqref{y-z} and \eqref{eq:gu}, respectively.  
\begin{theorem}\label{eq:MainTheorem}
Assume that $\lambda>0$ and $\varrho, \phi \geq 0$. Then, there exists a unique optimal trading speed $u^* \in \mathcal A$ with  corresponding controlled {trajectories $Y^{u^*}$ and}  $g^{u^*}$ such that 
	\begin{align}	\label{eq:optimalcontrol_monotone}
	u^*_t &= \frac 1 {2\lambda} \big(  \E[ (P_t - P_T) \mid \mathcal F_t]  - Y^{u^*}_t + \langle \Theta_t , K_t \rangle_{L^2} +  \langle \boldsymbol {\Gamma}_t^{-1} K_t  , g_t^{u^*} \rangle_{L^2} \big),
	\end{align}
	for all $t \leq T$. 
Moreover, the optimal value process is given by 
$$
V_t^{*}=  \frac{1}{2}\left( \langle g_t^{u^*}, \boldsymbol {\Gamma}_t^{-1} g_t^{u^*} \rangle_{L^2} + 2\langle \Theta_t, g_t^{u^*} \rangle_{L^2} + 2 \mathbb E[P_T \mid \mathcal F_t] Q_t^{u^*} +  \chi_t \right).
$$
	\end{theorem}
The proof of Theorem \ref{eq:MainTheorem} is given in Section \ref{sec-thm-pf}. 

In the following proposition we rewrite the optimizer $u^*$, which is given in a feedback form in \eqref{eq:optimalcontrol_monotone}, in an explicit form after observing that the linearity of the process $g^u$ in $u$, yields that $u^*$ in \eqref{eq:optimalcontrol_monotone}  solves the linear  Volterra  equation
\be \label{volt-u} 
 u^*_t = a_t + \int_0^t B(t,s)u_s^*ds,
 \ee
with the process $\{a_t\}_{t\in [0,T]}$ and the kernel $B$ which are given by
	\begin{equation} \label{eq:aB}
	\begin{aligned} 
a_t &:= \frac 1{2\lambda} \left(\E\left[ (P_t - P_T) \mid \mathcal F_t\right]  - \tilde h_0(t) + \langle \Theta_t , K_t \rangle_{L^2} +  \langle \boldsymbol {\Gamma}_t^{-1} K_t  , \mathds 1_t  ( \tilde h_0, q )^\top \rangle_{L^2} \right), \\
B(t,s) &:=  \mathds 1_{\{s<t\}}\frac 1{2\lambda} \left( \langle \boldsymbol {\Gamma}_t^{-1} K_t  , \mathds 1_t  ( \tilde G(\cdot,s), -1 )^\top \rangle_{L^2}   - \tilde G(t,s) \right). 
	\end{aligned}  
	\end{equation}
\begin{proposition} \label{prop-exp-u} 
Assume that $\lambda>0$ and $\varrho, \phi \geq 0$. Then the maximizer of \eqref{def:objective2}, $u^*$ is given by 
$$
    u_t^* = \left((\id - \boldsymbol B)^{-1} a \right) (t), \quad 0\leq t \leq T,  
$$
with $a$ given in \eqref{eq:aB} and $\boldsymbol B$ is the integral operator induced by the kernel $B$ in \eqref{eq:aB}.
\end{proposition} 
The proof of Proposition \ref{prop-exp-u} is given in Section \ref{sec-prop-expl}. 

\blue{
\begin{remark}
Note that the optimal strategy in Proposition \ref{prop-exp-u} is the continuous-time analog to the discrete-time solutions of \citet{AlmgrenChriss1,OPTEXECAC00}, in the special case where $P$ in \eqref{ass:P} is a martingale and the propagator in \eqref{z-def} is $G\equiv 0$. For the  continuous-time version of the aforementioned papers we refer to Chapter of 6.4 of \cite{cartea15book}.\end{remark} }

 \begin{remark} 
 In \cite{NeumanVoss:20} the special case of an exponentially decaying transient price impact of the form $  G (t,s) = \mathds 1_{\{s<t\}} e^{-\lambda (t-s)}$ was considered, along with a semimartingale unaffected price process $P$. The finite variation component $A=\{A_t\}_{t\geq 0}$ of $P$ was interpreted as a price predictive signal observed by the trader. Here we are considering a general Volterra kernel $G$ and the signal takes a more general form as $A_t=\E[P_t-P_T| \mathcal F_t]$ for $P$ progressively measurable.  
\end{remark} 

\begin{remark} 
Theorem \ref{eq:MainTheorem} and Proposition \ref{prop-exp-u} extend the results of \citet{Alf-Schied-13} in a few directions. In contrast to \cite{Alf-Schied-13}, we assume that the price process $P$ is progressively measurable and not necessarily a martingale, which turns the control problem from deterministic to stochastic optimisation, and introduces new ingredients in the value function \eqref{eq:candidatevaluefun}, such as $L^2$-valued free-boundary BSDE (see \eqref{eq:candidatevaluefun} and \eqref{prop-theta}) and linear BSDE \eqref{eq:BSDEchi}. Moreover, it is assumed in \cite{Alf-Schied-13} that $G$ is a convolution kernel which is completely monotone and satisfies $G''(+0)<\infty$. This is a special case of assumption \eqref{pos-def} as implied by Example 2.7
in \cite{GSS}. Here we remove these restricting assumptions, which allows us to consider power law kernels of the form $G(t) =t^{-\alpha}$ for $0<\alpha<1/2$ and non-convolution kernels as in Remark \ref{rem-examples}. Lastly, we incorporate a risk aversion term into the cost functional \eqref{eq:optimalvalue_monotone}, which has an important practical role as it reflects the risk of holding inventory.
\end{remark}

 \begin{remark}
 The solution to the problem in \cite{Alf-Schied-13} is given in terms of an infinite dimensional Riccati equation which takes values in $\mathbb{R}$ (see eq. (5) and (6) therein). More generally, the Riccati equations of \cite{Alf-Schied-13} appear in the context of linear-quadratic stochastic Volterra control problems for the specific case of convolution kernels that admit a representation as Laplace transforms of certain measures, see \cite{abi2021integral,abi2021linear}. This could be compared with our operator-valued Riccati equation in \eqref{eq:riccati_psiBold} which is one of the main ingredients of the solution (see \eqref{eq:optimalcontrol_monotone} and \eqref{psi-gamma}) and which is valid for a larger class of kernels.  More precisely, in the specific case where \( G \) takes a convolution form  with  a completely monotone function as in Lemma~\ref{lemma-pos-def-ker}, then introducing the Markovian auxiliary variables \( \tilde Z^u_t(\rho) := \int_0^t e^{-\rho(t-s)} \, u_s \, ds \), for $\rho \in \mathbb R_+$, allows us to link our state variable \( g_t(s) \) in \eqref{eq:gu} with the family \( (\tilde Z_t(\rho))_{\rho \in \mathbb{R}_+} \), as detailed in \cite[Lemma 3.5]{abi2022laplace}. In this setup, our value function can be re-expressed as linear-quadratic  in \( (\tilde Z, Q) \), involving the term
\[
\int_{\mathbb{R}_+^2} \tilde Z_t(\rho) \, \Lambda_t(\rho, \rho') \, \tilde Z_t(\rho') \, \mu(d\rho) \, \mu(d\rho'),
\]
where the measure $\mu$ comes from \eqref{g-spec} and 
\[
\Lambda_t(\rho, \rho') := \int_t^T e^{-\rho'(s-t)} \left(\boldsymbol{D}_t^{-1}\right) \left(e^{-\rho (\cdot -t)} {1}_{\{t<\cdot\} }\right)(s) \, ds.\] Using the operator Riccati equation satisfied by \( \boldsymbol{D}^{-1}_t \), we can show, following similar computations as in \cite[Proposition 3.7]{abi2022laplace}, that \( \Lambda \) solves an infinite-dimensional Riccati equation, in line with the equations in [6].

As stated in Section 1.3 of  \cite{Alf-Schied-13} their function-valued Riccati equation in general cannot be solved explicitly. Only one tractable example is provided for the case where $G$ is a finite sum of exponential kernels. In Theorem \ref{eq:MainTheorem} and Proposition \ref{prop-exp-u}, we provide an explicit solution to the problem. In Section \ref{sec-simulation}, we show that our formulas can be implemented in a straightforward and efficient way for a large class of price impact kernels. In particular, our results cover the case of non-convolution singular price impact kernels such as the power-law kernel (see Remark \ref{rem-examples} for additional examples).

 \end{remark}
 
\section{Numerical illustration} \label{sec-simulation} 
In this section, we provide an efficient  numerical discretization scheme for the optimal trading speed $u^*$ in \eqref{eq:optimalcontrol_monotone}. We then  illustrate numerically the effect of the transient impact kernel $G$ and the signal on the optimal trading speed.  For simplicity, we will fix throughout this section  the penalization on the running inventory to zero, i.e.~$\phi=0$ in \eqref{def:objective}. The code of our implementation can be found at \url{https://colab.research.google.com/drive/1VQasI92YhdBC0wnn_LxMkkx_45VyK1yQ}.

\subsection{Discretization of the operators}
We will make use of the so-called Nyström method to discretize the following integral equation for $u^*$
$$ u^*_t = a_t + \int_0^t B(t,s)u_s^*ds, \quad t\in [0,T],$$
recall \eqref{volt-u}, where $a$ and $B$ are given by \eqref{eq:aB}.

Fix $n\in \mathbb N$ and a partition $0=t_0<t_1<t_2<\ldots<t_n=T$ of $[0,T]$. A discretization of the equation for $u^*$ leads to the approximation of the values  $(u^*_{t_i})_{i=0,\ldots, n}$ by the vector $u^{(n)} \in \R^{n+1}$ given by
\begin{align}
u^{(n)}:= (I_{n+1}-B^{(n)})^{-1} a^{(n)}, \label{eq:approxun}
\end{align} 
with  $a^{(n)} \in \R^{n+1}$ and $B^{(n)} \in \R^{(n+1)\times (n+1)}$ given  by\footnote{We note that indices count for vectors and matrices start from $0$.}
\begin{align*}
 a^{(n)} &:=(a_{t_0}, a_{t_1}, \ldots, a_{t_n})^\top,\\
B^{(n)}_{ij} &:= 1_{\{j \leq i-1\}} \int_{t_j}^{t_{j+1}}B(t_i,s) ds, \quad i,j=0,1,\ldots, n. 
\end{align*}

We now provide a detailed approximation for $a^{(n)}$ and $B^{(n)}$ for the case $\phi=0$.  We start by defining the only quantities that depend on the signal $P$ and the kernel  $G$ that need to be (pre)computed for the approximations.  First, we denote by {$\nu_t$} the  following conditional expectation 
\begin{align}\label{eq:approxR}
\nu_t(s) :=  1_{\{s\geq t\}} \mathbb E[P_s-P_T|\mathcal F_t], \quad s,t \in [0,T],
\end{align}
and by $\mathbf{N}$ the following $(n+1)\times(n+1)$-matrix:
\begin{align}\label{eq:approxRn}
\mathbf{N}^{kj} &:= \nu_{t_j}(t_k), \quad k,j=0,\ldots, n.
\end{align}
Second, we define the following $(n+1)\times (n+1)$  lower and upper triangular matrices $L$ and $U$ where the non-zero elements are given by:
\begin{align}
    L^{kj} &:= \int_{t_j}^{t_{j+1}}   \tilde G(t_k,s) ds, \quad k =0, \ldots, n, \quad j=0,\ldots, (k-1), \label{eq:approxL}\\
 U^{kj} &:=  \int_{t_j}^{t_{j+1}} \tilde G(s,t_k) ds, \quad k=0, \ldots, n, \quad j=k, \ldots, (n-1)\label{eq:approxU},
\end{align}
where $\tilde G$ was defined in \eqref{h-tilde}. 

\paragraph{Step 1. Discretization of $\langle \boldsymbol {\Gamma}_{t_i}^{-1} K_{t_i}  ,  1_{t_i}  ( f, g )^\top \rangle_{L^2}$.}
Fix  $i=0,\ldots,n$ and $f,g\in L^2([0,T],\R)$. We first look at approximating the term $ \langle \boldsymbol {\Gamma}_{t_i}^{-1} K_{t_i}  ,  1_{t_i}  ( f, g )^\top \rangle_{L^2}$ from \eqref{eq:aB}. We note that the expressions simplify  for the case $\phi = 0$ (see \eqref{op-gam-inv}), so that using the fact that $\boldsymbol D_t$ is self-adjoint, we obtain that
\be\label{eq:approxinner}  
\begin{aligned}
 \langle \boldsymbol {\Gamma}_{t_i}^{-1} K_{t_i}  ,  1_{t_i}  ( f, g )^\top \rangle_{L^2}  &=  \langle \boldsymbol 1_{t_i} \tilde G_{t_i},
 \boldsymbol{D}_{t_i}^{-1} f \rangle_{L^2} \nonumber\\
 &= \int_{t_i}^T \tilde G(s,t) (\boldsymbol{D}_{t_i}^{-1} f)(s) ds \nonumber \\
 &\approx \sum_{k=i}^{n-1} \int_{t_k}^{t_{k+1}} \tilde G(s,t_i) ds (\boldsymbol{D}_{t_i}^{-1}f)(t_k). 
\end{aligned}
\ee
The action of the operator $\boldsymbol D_{t_i}$ can be  approximated by the $n\times n$ matrix  $D^{(n)}_{t_i}$ defined by 
\begin{align}
D^{(n)}_{t_i}  &:= 2 \lambda I_n + d^{(n)}_{t_i}, \label{eq:approxDn} \\
d^{(n),kj}_{t_i} &:= L^{kj} 1_{\{i \leq j \leq (n-1)\}} + U^{kj} 1_{\{i \leq k \leq (n-1)\}}, \quad k,j = 0,\ldots, n-1. \label{eq:approxDn2}
\end{align}
Combining this with  \eqref{eq:approxinner} yields the approximation 
\be\label{eq:approxinner2}
\begin{aligned}
 \langle \boldsymbol {\Gamma}_{t_i}^{-1} K_{t_i}  ,  1_{t_i}  ( f, g )^\top \rangle_{L^2} &\approx \sum_{k=i}^{n-1} \int_{t_k}^{t_{k+1}} \tilde G(s,t_i) ds ((D^{(n)}_{t_i})^{-1} f^{(n)})(t_k)  \\
 &=  U_i^\top (D^{(n)}_{t_i})^{-1} f^{(n)},
\end{aligned}
\ee
 where $f^{(n)}:=(f(t_0), f(t_1),\ldots, f(t_{n-1}))^\top$ and $U_i:=(U^{i\, 0},U^{i\,1}, \ldots U^{i \, (n-1)})$, i.e.~the $n$-dimensional $i$-th row of $U$ excluding the last term.

\paragraph{Step 2. Discretization of $B^{(n)}$.}
For $i=0,\ldots,n$ and $j=0,\ldots,(i-1)$, it follows that 
\begin{align}
B^{(n)}_{ij} &=  \int_{t_j}^{t_{j+1}}B(t_i,s) ds  \nonumber \\
&= \frac{1}{2\lambda}     \langle \boldsymbol 1_{t_i}\tilde G_{t_i}  ,  \boldsymbol {D}_{t_i}^{-1} \left(  \int_{t_j}^{t_{j+1}}   \tilde G(\cdot,s) ds \right) \rangle_{L^2}    - \frac{1}{2\lambda} \int_{t_j}^{t_{j+1}}   \tilde G(t_i,s) ds \nonumber \\
&\approx  \frac{1}{2\lambda}    (U_i)^\top (D^{(n)}_{t_i})^{-1} L^{(j)}   - \frac{1}{2\lambda} L^{ij},    \quad i,j=0,\ldots, n, \label{eq:approxBn}
\end{align}
where we used \eqref{eq:approxinner2} for the last identity and $L^{(j)}:=(L^{0\, j},L^{1\, j}, \ldots L^{(n-1)\, j})^\top$, i.e.~the $j$-th column of $L$ excluding the last element. 

\paragraph{Step 3. Discretization of $a^{(n)}$.}  Fix $i=0,\ldots,n$.
Recall from \eqref{th-sol} and \eqref{eq:approxR} that 
$$\Theta_t(s)=- \left(\boldsymbol {\Gamma}^{-1}_t \nu_t e_1 \right)(s),
$$
so that using \eqref{eq:approxinner2} we obtain
\begin{align*}
\langle  K_{t_i}, \Theta_{t_i} \rangle_{L^2} =  - \langle \boldsymbol {\Gamma}^{-1}_{t_i} K_{t_i},   \nu_{t_i} e_1   \rangle_{L^2} 
 \approx - (U_i)^\top (D^{(n)}_{t_i})^{-1}  \mathbf{N}^{i}
\end{align*}
where  $ \mathbf N^i:=(\mathbf N^{0\,i},\mathbf N^{i\,1},\ldots, \mathbf N^{i\,(n-1)})$, i.e.~the $i$-th column of $\mathbf N$ defined in \eqref{eq:approxR} excluding the last term. Another application of \eqref{eq:approxinner2} yields the following approximation for $a_{t_i}$:
\be\label{eq:approxan}
\begin{aligned}
a_{t_i} &\approx \frac 1{2\lambda} \left(\mathbf{N}^{i\,i}  - \tilde h_0(t_i)  - (U_i)^\top (D^{(n)}_{t_i})^{-1} \mathbf{N}^i +  (U_i)^\top (D^{(n)}_{t_i})^{-1} \tilde h_0^n \right), \quad i=0,\ldots, n,
\end{aligned}
\ee
with $\tilde h_0^n=(\tilde h_0(t_0), \tilde h_0(t_1),\ldots, \tilde h_0(t_{n-1}))^\top$.

\paragraph{Summary.} To sum up, the implementation is straightforward:

\fbox{\parbox{\textwidth}{\begin{enumerate}
    \item Specify the signal $P$ and the kernel $G$ as inputs and compute the $(n+1)\times (n+1)$ matrices $N$, $L$ and $U$ using  \eqref{eq:approxRn}, \eqref{eq:approxL} and \eqref{eq:approxU}.  (Refer to  Subsection~\ref{s:numsub} below for explicit examples.)
    \item Construct the $n\times n$-matrices $D^{(n)}_{t_i}$ using \eqref{eq:approxDn}-\eqref{eq:approxDn2} for $i=0,\ldots, n$. 
    \item 
Construct the $(n+1)$-vector $a^{(n)}$ using \eqref{eq:approxan} and the $(n+1)\times(n+1)$ matrix  $B^{(n)}$ using \eqref{eq:approxBn}.
\item 
Recover the $(n+1)$-vector for the optimal control path $u^{(n)}$ from \eqref{eq:approxun}. 
\end{enumerate}
}}

\subsection{Numerical examples}\label{s:numsub}
For our numerical illustrations, we fix a uniform partition with mesh size $\Delta t:=T/n$ and we consider a signal of the form
\begin{align*}
  P_t = \int_0^t I_s ds + M_t,
\end{align*}
for some martingale $M$ with $I$ an Ornstein-Uhlenbeck process of the form
\be \label{ou} 
dI_t = - \gamma I_t dt + \sigma dW_t, \quad I_0 \in \mathbb R,
\ee
where $\gamma, \sigma$ are positive constants and $W$ is a Brownian motion. In this case, the conditional expectation process $\nu_t$ given in \eqref{eq:approxR}  can be computed explicitly:
\begin{align*}
\nu_t(s) = I_t \frac{e^{-\gamma(T-t)}-e^{-\gamma(s-t)}}{\gamma} 1_{\{s\geq t\}},
\end{align*}
so that $\mathbf N$ defined in \eqref{eq:approxRn} reads 
\begin{align*}
\mathbf N^{kj} =  I_{t_j} \frac{e^{-\gamma(n-j)\Delta t}-e^{-\gamma({k}-j)\Delta t}}{\gamma} 1_{\{k\geq j\}}, \quad k,j=0, \ldots, n.
\end{align*}

 We will consider two examples of transient impact convolution kernels from Remark \ref{rem-examples}: the exponential kernel and the power-law kernel, where for computational convenience we take $\beta = 1-\alpha$ with $\alpha\in(1/2,1)$, for the exponent of the power law (see Table \ref{T:kernels}). The results will be compared with the case of no transient impact, i.e. $G\equiv0$.  In all three cases the matrices $L$ and $U$ in \eqref{eq:approxL}-\eqref{eq:approxU} can be computed explicitly and are also given in Table \ref{T:kernels} below.
\begin{table}[h!]
\centering
\resizebox{\textwidth}{!}{
\begin{tabular}{c c c c }
\hline\hline
&  & $L^{kj}$  & $U^{kj}$ \\ 
& $G(t,s)$ & for $0\leq j \leq k-1$ & for $k\leq j \leq n-1$ \\
\hline \hline \\[0.5ex]
No-transient		& 0 & $2 \varrho \Delta t$ & $2 \varrho \Delta t$\\ \\
Exponential	& $c{\rm e}^{-\rho(t-s)}\mathds 1_{\{s<t\}}$ & $ 2 \varrho \Delta t + c \displaystyle\frac{e^{\rho \Delta t}- 1}{\rho} e^{-\rho(k-j)\Delta t}$ & {$ 2 \varrho \Delta t + c\displaystyle\frac{1- e^{-\rho \Delta t}}{\rho}e^{-\rho(j-k)\Delta t}$}\\ \\
Fractional		& $c\,{(t-s)^{\alpha-1}} \mathds 1_{\{s<t\}}$ & $ 2 \varrho \Delta t + \frac{c (\Delta t)^{\alpha}}{\alpha}((k-j)^{\alpha} - (k-j-1)^{\alpha}  )$ & $2 \varrho \Delta t + \frac{c (\Delta t)^{\alpha}}{\alpha}((j+1-k)^{\alpha} - (j-k)^{\alpha}  )$\\ \\
\hline
\end{tabular}
}
\caption{Some kernels $G$ and the corresponding explicit non-zero elements of the matrices $L$ and $U$ in \eqref{eq:approxL}-\eqref{eq:approxU}.}
\label{T:kernels}
\end{table}

In Figure \ref{fig:withoutsignal} we present the optimal trading speed in the left panel and the resulting inventory in the right panel, in the absence of a signal (i.e. $I=0$), where the parameters of the model are set to 
\be
h_0 \equiv 0,  \ q_0 = 10,
\ T = 10, \ \lambda = 0.5, \ \varrho = 4, \  \phi =0, \ \rho=0.5, \ \alpha = 0.55, \  c=1.
\ee
We consider the cases where $G\equiv 0$ (blue), $G(t,s) =e^{-\rho(t-s)}\mathds 1_{\{s<t\}} $ with $\rho=0.5$ (orange) and $G(t,s) =(t-s)^{\alpha-1}\mathds 1_{\{s<t\}} $ with  $\alpha = 0.55$ (green). 
We notice that the optimal strategy in the power law case is more restrained than the one of the exponential kernel, as the transient price impact resulting by trades has a slower decay. This effect becomes even more prominent when we incorporate a trading signal in Figure \ref{fig:withsignal}. 

	\begin{figure}
		\includegraphics[scale=.55,center]{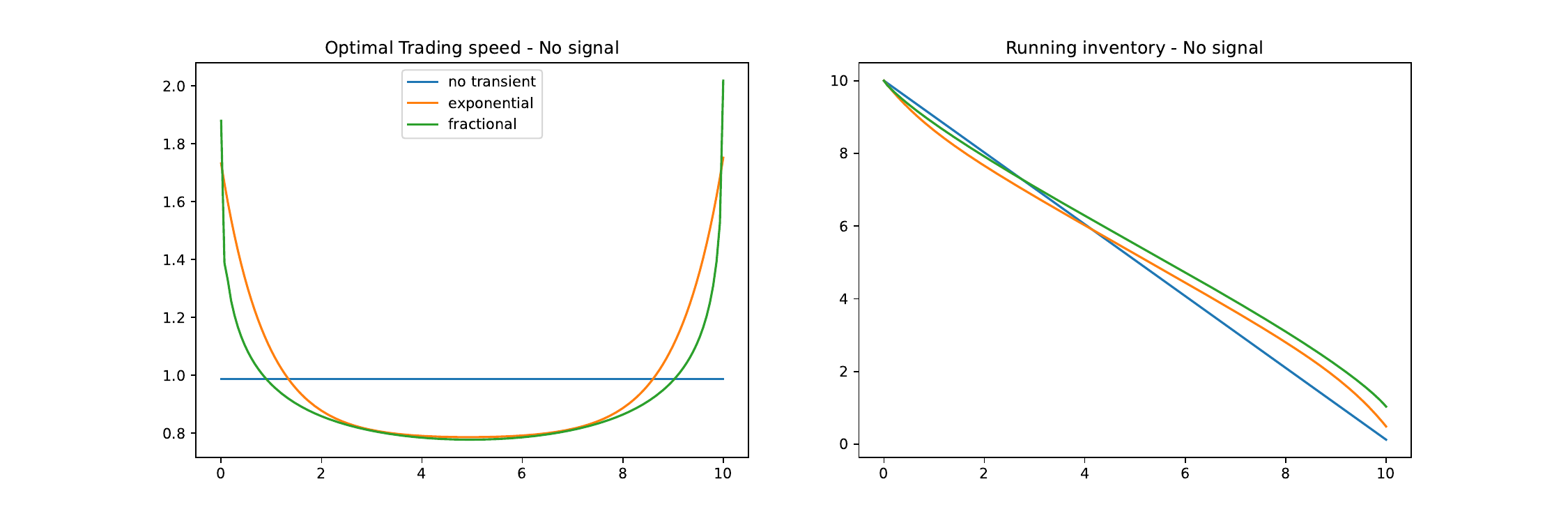}
				\vspace{-1cm}
		\captionof{figure}{Impact of different kernels on the optimal trading speed and inventory in the absence of a signal for the parameters  $h_0 \equiv 0, q_0 = 10,
T = 10, \lambda = 0.5, \varrho = 4, \phi =0,$ with three impact kernels: $G\equiv 0$ (blue), $G(t,s) =e^{-\rho(t-s)}\mathds 1_{\{s<t\}} $ with $\rho=0.5$ (orange) and $G(t,s) =(t-s)^{\alpha-1}\mathds 1_{\{s<t\}} $ with  $\alpha = 0.55$ (green). }
		\label{fig:withoutsignal}
	\end{figure}

In Figure \ref{fig:withsignal} we plot the optimal strategy for an agent who is executing a sell strategy and is also observing an integrated  Ornstein-Uhlenbeck signal as in \eqref{ou} with parameters $I_0 = \pm 2, \gamma = 0.3, \sigma = 0.5$. When the signal is negative, as illustrated in the upper panels, the agent trades with an excessive speed in the exponential kernel case compared to the power law case. 
This difference is not as substantial for a positive signal as in this scenario the trader is trading slowly anyway, as the value of her portfolio will increase in the immediate future due to the effect of the signal. Since the trading in the positive signal case is slow at the beginning of trade, the strategy is less sensitive to the type of price impact kernel. Towards the end of the trading period, inventory penalties become more influential and they trigger rapid sells, so again the effect of the kernel type is not significant. In Figure \ref{fig:withsignal2} the transient price impact resulting by the optimal strategies for the cases of exponential and power law kernels is presented, with the same realization of the signal as in Figure \ref{fig:withsignal} are used. One can observe in Figure \ref{fig:withsignal2} that the price impact induced by the power law kernel is significantly more persistent than in the exponential kernel case.   

	\begin{figure}
		\includegraphics[scale=.6, center]{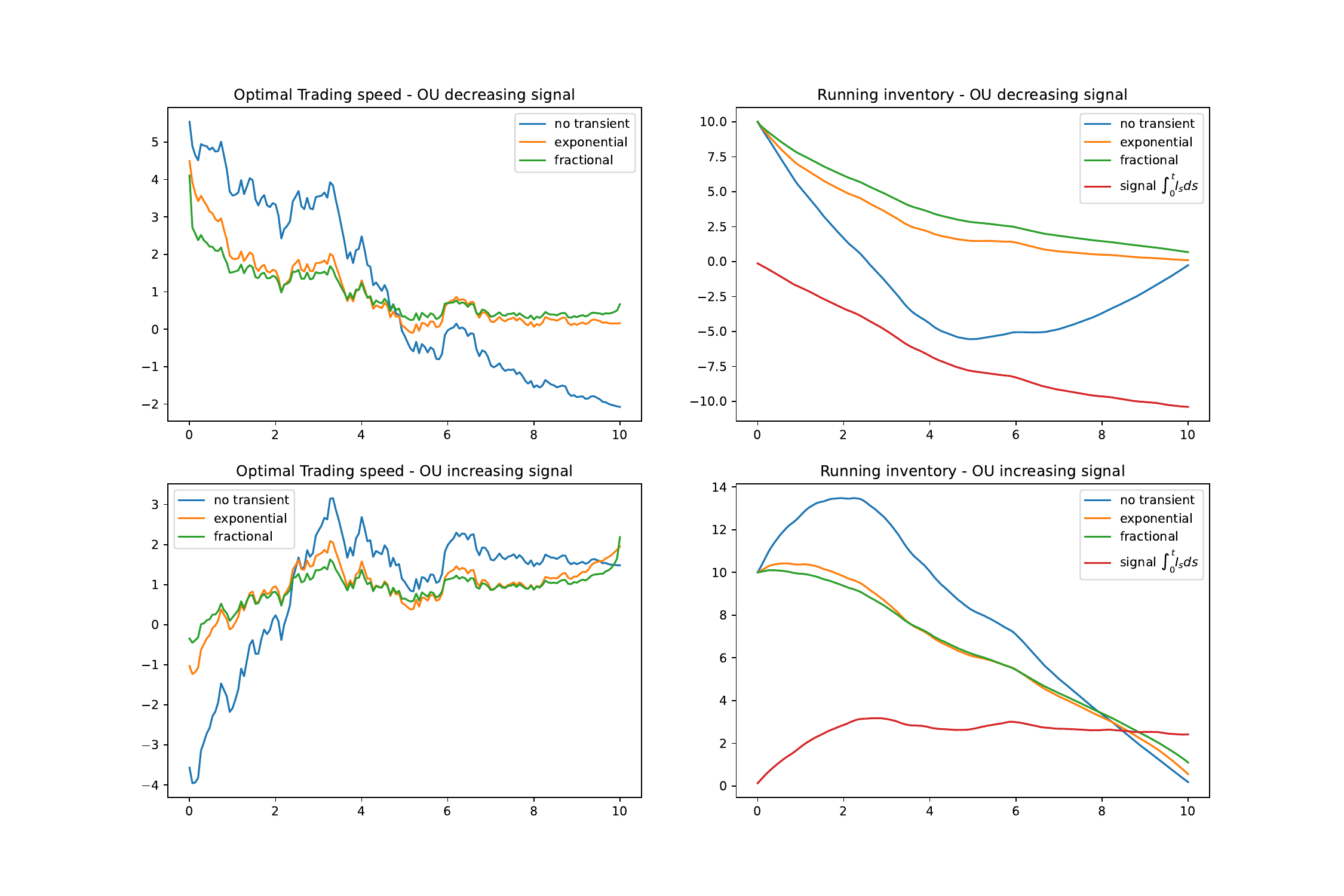}
				\vspace{-2 cm}
		\captionof{figure}{Effect of different kernels on the optimal trading speed and inventory in the presence of a signal, for the parameters  $h_0 \equiv 0, q_0 = 10,
T = 10, \lambda = 0.5, \varrho = 4, \phi =0,$ for the Ornstein-Uhlenbeck signal: $I_0 = -2, \gamma = 0.3, 
\sigma = 0.5$ (upper panels) and $I_0 = 2, \gamma = 0.3, 
\sigma = 0.5$  (lower panels) and with three price impact kernels: $G\equiv 0$ (in blue), $G(t,s) =e^{-\rho(t-s)}\mathds 1_{\{s<t\}} $ with $\rho=0.5$ (orange) and $G(t,s) =(t-s)^{\alpha-1}\mathds 1_{\{s<t\}} $ with  $\alpha = 0.55$ (green). }
		\label{fig:withsignal}
\end{figure}

	\begin{figure}
		\includegraphics[scale=.6, center, trim={0 0 0 24cm},clip]{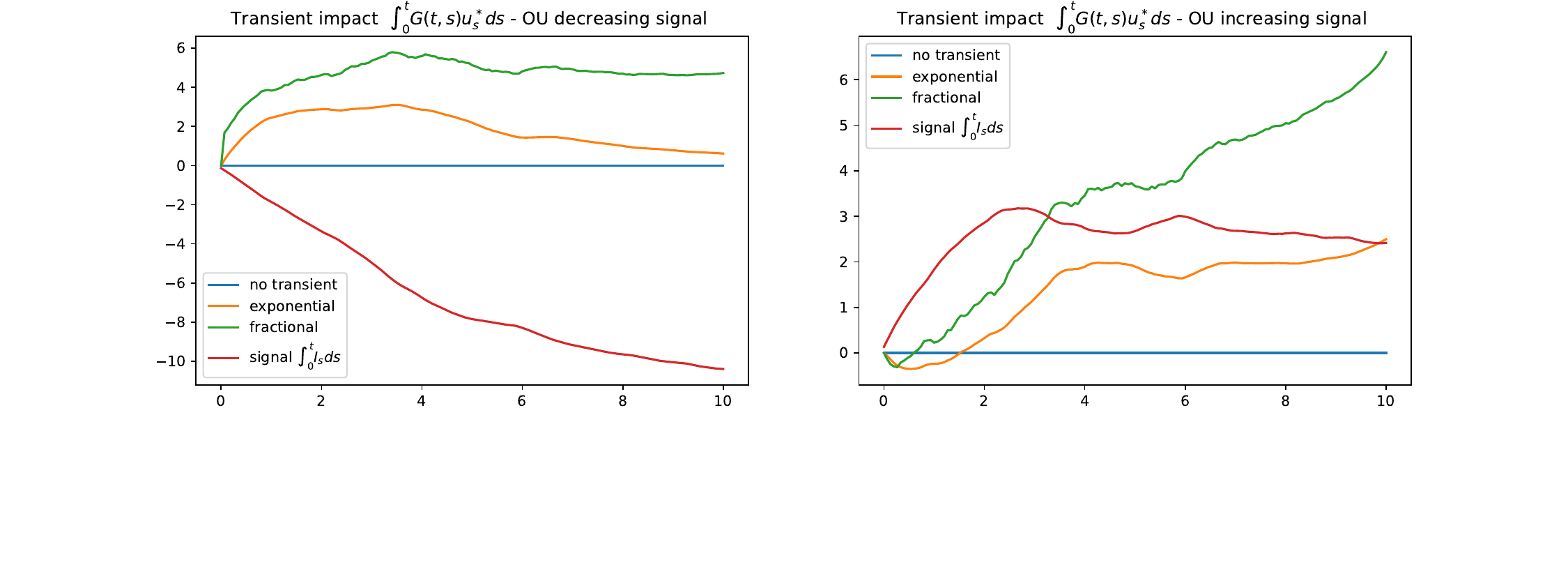}
		\vspace{-3cm}
		\captionof{figure}{The transient price impact of different kernels in the presence of a signal, for the parameters  $h_0 \equiv 0, q_0 = 10,
T = 10, \lambda = 0.5, \varrho = 4, \phi =0,$ and for Ornstein-Uhlenbeck signal (in red) with $\gamma = 0.3, \sigma = 0.5$, where $I_0 = -2$ in the left panel and $I_0 = 2$ in the right panel. The impact of the kernels appears for $G\equiv 0$ (blue), $G(t,s) =e^{-\rho(t-s)}\mathds 1_{\{s<t\}} $ with $\rho=0.5$ (orange) and $G(t,s) =(t-s)^{\alpha-1}\mathds 1_{\{s<t\}} $ with  $\alpha = 0.55$ (green). }
		\label{fig:withsignal2}
\end{figure}

In Figure \ref{fig:sensitivitykernelwosignal} we provide a sensitivity analysis for the optimal trading speed and the optimal inventory subject to changes in the price impact kernel parameters. In the left panels we consider fractional kernels and in the right panels we consider exponential kernels. For the factional kernel case, we observe that for small values of $\alpha$ , the kernel $t\mapsto  t^{\alpha-1}$ induces more price impact over small time intervals enforcing the agent to trade slower. On the other hand when $\rho$ increases the price impact induced by kernel $t\mapsto e^{-\rho t}$ decays faster which allows the agent to trade faster.  

In Figure \ref{fig:sensitivitykernelswithsignal} we repeat the same experiment, only now we extend the time horizon from $T=1$ to $T=10$  and add a positive Ornstein-Uhlenbeck signal similar to the one in Figure \ref{fig:withsignal}. We notice that the monotonicity with respect to the $\alpha$ parameter in the fractional kernel is preserved (see in the left panels). However in this scenario, since the signal is positive the agent is first buying in order to make a quick profit and then selling her inventory in order to close the position. We observe that larger values of $\alpha$ allow the trader to buy more inventory at the beginning of the trade.    
\begin{figure}
		\includegraphics[scale=.6, center]{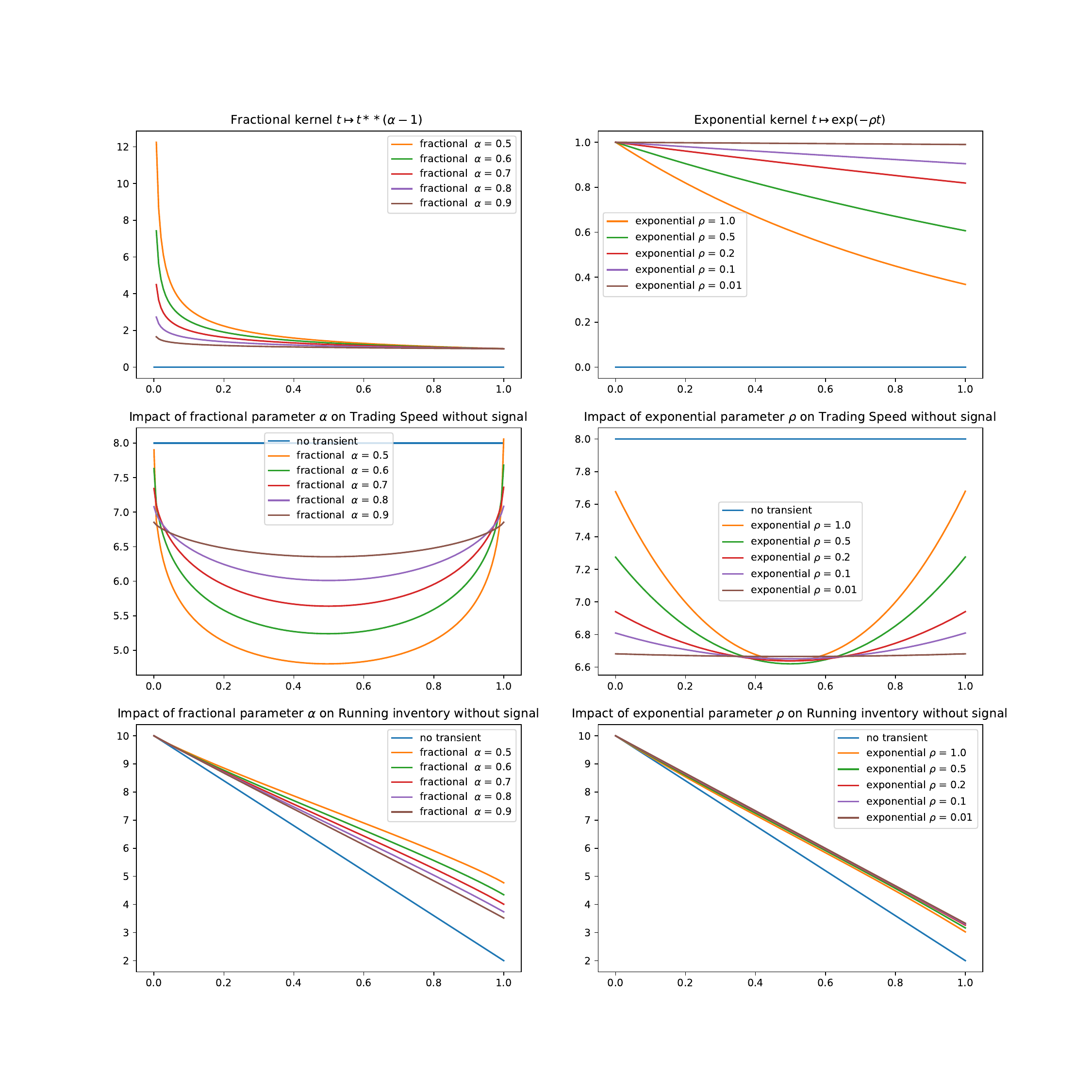}
		\vspace{-3cm}
		\captionof{figure}{Impact of parameters of the kernels on the optimal trading speed and inventory without signal for the parameters  $h_0 \equiv 0, q_0 = 10,
T = 1, \lambda = 0.5, \varrho = 2, \phi =0$. First column: Fractional kernel; Second column: Exponential kernel. }
		\label{fig:sensitivitykernelwosignal}
	\end{figure}



\begin{figure}
		\includegraphics[scale=.6, center]{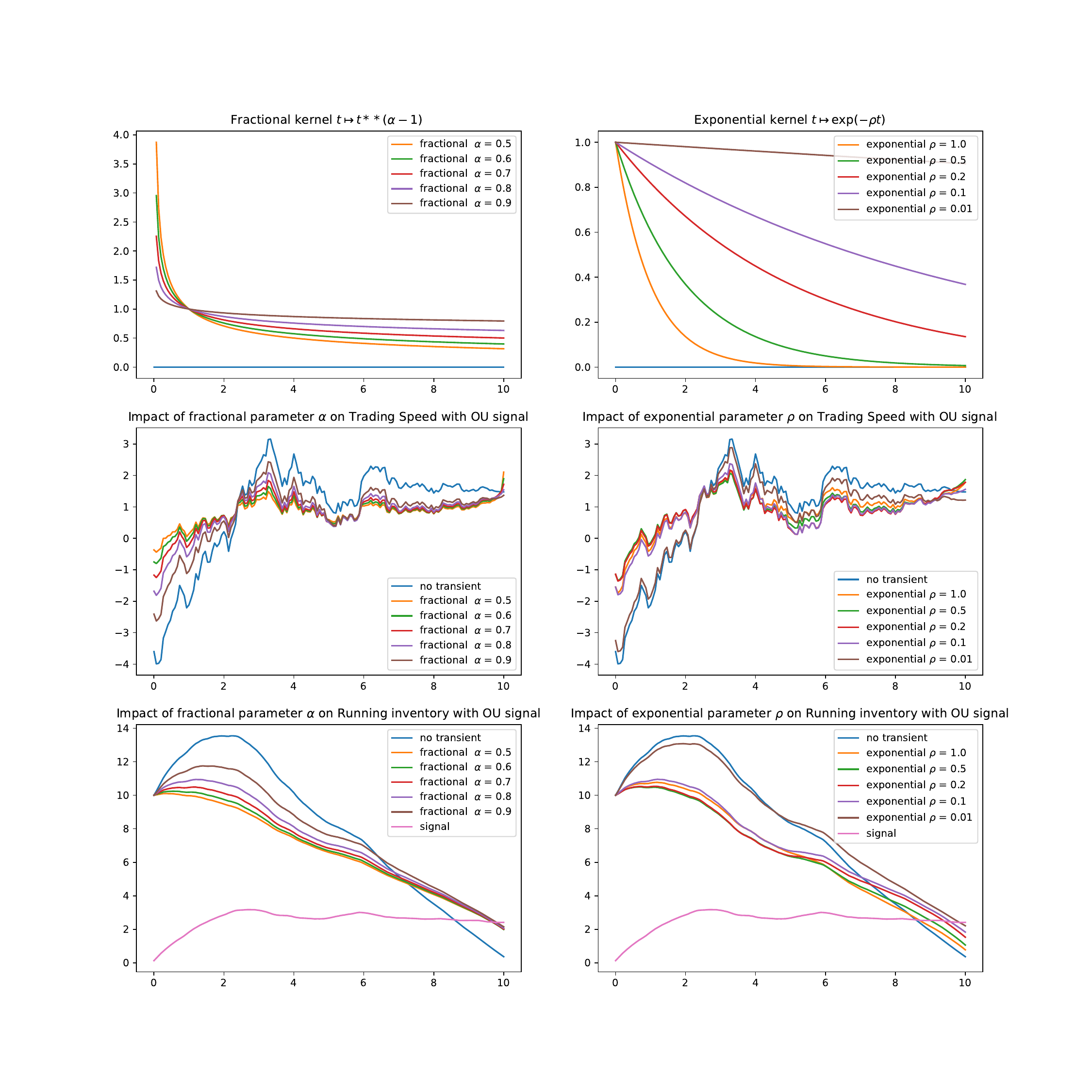}
				\vspace{-3.cm}
		\captionof{figure}{Impact of parameters of the kernels on the optimal trading speed and inventory with Ornstein-Uhlenbeck signal for the parameters   $h_0 \equiv 0, q_0 = 10,
T = 10, \lambda = 0.5, \varrho = 2, \phi =0$; for the OU signal: $I_0 = 2, \gamma = 0.3, \sigma = 0.5$. First column: Fractional kernel; Second column: Exponential kernel.}
		\label{fig:sensitivitykernelswithsignal}
	\end{figure}

\section{Derivation of the solution} \label{sec-thm-pf} 
\subsection{The covariance operator} 
 We define $\boldsymbol{\Sigma}_t$ the \emph{covariance operator} induced by $K$ (recall \eqref{g0-k-def}),  as the integral operator associated with the following kernel: 
\begin{align}\label{eq:sigmakernel}
{\Sigma}_t(s,u) :=\frac 1 {4\lambda} \int_t^{s\wedge u} K(s, z)K^{\top}(u, z)   dz, \quad t \leq s,u \leq T,
\end{align} 
where we recall that $\lambda$ is as in \eqref{def:S}.
Let $\id$ denote the identity operator, i.e. $(\id f)=f$ for all $f \in L^2\left([0,T],\R^2 \right)$. 

We define $\boldsymbol{\hat K}$ as the integral operator induced by the kernel $K$ as follows, 
\begin{align} \label{K-def} 
\hat K(t,s)  :=  - K(t,s) \otimes e_1, \quad \mbox{with } e_1=(1, 0)^\top,
\end{align}
where $\otimes$ represents the outer product. 
 Specifically, we have 
\be \label{eq:explicithatK}
\hat K(t,s)=  \begin{pmatrix}
-\tilde G(t,s) &  0 \\
\mathds{1}_{\{s\leq t\}}  & 0  \\
 \end{pmatrix}.
\ee

We define 
the \emph{adjusted covariance integral operator} $\boldsymbol{\tilde{\Sigma}}_t$ as follows 
\begin{align}\label{def:C_tilde}
\boldsymbol{\tilde{\Sigma}}_t =   \left( \id - \frac 1 {2\lambda} \boldsymbol{\hat K} \right)^{-1} \boldsymbol{\Sigma}_t \left(\id - \frac 1 {2\lambda} \boldsymbol{\hat K}^*\right)^{-1}.
\end{align}
Note that Lemma A.5 in \cite{aj-22} ensures that  $\left( \id - \frac 1 {2\lambda} \boldsymbol{\hat K} \right)$ and $\left(\id - \frac 1 {2\lambda} \boldsymbol{\hat K}^*\right)$ are invertible.

\subsection{The Riccati operator}
We let    
\be \label{a-mat} 
A := \left(\begin{matrix}
\frac 1 {2\lambda} & 0 \\
0 & - 2\phi
\end{matrix}\right).
\ee
Recall that $\boldsymbol{\tilde{\Sigma}}_t$ was defined in \eqref{def:C_tilde}. We define 
\begin{align}\label{def:riccati_operator}
\boldsymbol \Psi_{t}:=   \left( \id - \frac 1 {2\lambda} \boldsymbol{\hat K}^* \right)^{-1} A \left( \id -  2\boldsymbol{\tilde \Sigma}_{t} A  \right)^{-1}     \left( \id - \frac 1 {2\lambda}  \boldsymbol{\hat K} \right)^{-1},  \quad  t\leq T.
\end{align}
First we identify $ \boldsymbol {\Psi}_t$ with $\boldsymbol {\Gamma}^{-1}_t$ from \eqref{op-gam-inv}, on a certain class of test functions. Recall that the notation $\mathds{1}_t$ was introduced in \eqref{ind-def}. 
\begin{lemma} \label{lemma-gam-phi} 
The operator $\boldsymbol{\Psi}_{t}$ {is well defined} and satisfies for any $f \in L^2([0,T],\R^2)$, 
\begin{align}\label{psi-gamma}
 \mathds{1}_t(s) \big(\boldsymbol {\Psi}_t f \mathds 1_t \big)(s)  =   \mathds{1}_t(s) \big( \boldsymbol {\Gamma}^{-1}_t  f\mathds 1_t \big)(s) , \quad \textrm{for all } s,t\in [0,T].
\end{align}
\end{lemma} 
The proof of Lemma \ref{lemma-gam-phi} is given in Section \ref{sec-lem-inv}.

We will show that  $\boldsymbol{\Psi}_t$ is a solution to a Riccati operator equation involving the covariance operator $\boldsymbol{\Sigma}_t$ induced by the kernel \eqref{eq:sigmakernel}. 
 For this we specify our notion of differentiability: for any operator $\boldsymbol{G}$ from $ L^2\left([0,T],{\R^2}\right)$ to itself we define the operator norm, 
\be \label{op-norm} 
 \|\boldsymbol{G}\|_{\rm {op}}:= \sup_{f \in L^2([0,T],\R^2)} \frac{\|\bold G f\|_{L^2}}{\|f\|_{L^2}}.
\ee
The operator $t\mapsto \boldsymbol{\Psi}_t$ is said to be strongly differentiable at time $t\geq  0$,  if there exists a bounded linear operator $\dot{\boldsymbol{\Psi}}_t$  from $ L^2\left([0,T],{\R}\right)$  into  itself  such that
\begin{align}
\lim_{h\to 0} \frac{1}{h} \| \boldsymbol{\Psi}_{t+h}-\boldsymbol{\Psi}_{t} -h \dot{\boldsymbol{\Psi}}_{t} \|_{\rm{op}}=0. 
\end{align}
The following lemma gives some fundamental properties of $\boldsymbol{\Psi}_t$ that will be useful for the proof of Theorem \ref{eq:MainTheorem}. Recall that $A$ was defined in \eqref{a-mat} and $\boldsymbol{\hat {K}}$  was defined in \eqref{K-def}. 

\begin{lemma}\label{L:Psi} For any $0\leq t\leq T$, $\boldsymbol{\Psi}_t$ given by \eqref{def:riccati_operator}  is a bounded  linear operator from $L^2\left([0,T],\R\right)$ into itself. Moreover we have, 
	\begin{itemize}
		\item[\bf{(i)}] \label{L:Psi1}
		$\bar{\boldsymbol{\Psi}}_t=( - A\cdot \id + \boldsymbol{\Psi}_t)$ is an integral  operator induced by a symmetric kernel  $\bar\psi_t(s,u)$ that satisfies  
				\bd
		\label{L:Psi2} 
		\sup_{t\leq T} \int_{[0,T]^2}|\bar \psi_t(s,u)|^2 ds du<\infty.
		\ed
		\item [\bf{(ii)}] For any $f \in L^2\left([0,T],\R^2\right)$,   
			\bd
		\label{eq:Psi_on_boudary}
		(\boldsymbol{\Psi}_t f \mathds{1}_t)(t) =  \left( A\cdot\id + \frac{1}{2\lambda}\boldsymbol{\hat {K}}^*\boldsymbol{\Psi}_t \right)(f \mathds{1}_t)(t),
		\ed
		where $\mathds{1}_{t}(s) = \mathds{1}_{\{t\leq s\}}$.
		 \item [\bf{(iii)}]
		 \label{L:Psi3}
		$t\mapsto \boldsymbol{\Psi}_t$ is strongly differentiable and  satisfies the operator Riccati equation
		\begin{align}
		\label{eq:riccati_psiBold}
		\dot{\boldsymbol{\Psi}}_t &= 2\boldsymbol{\Psi}_t \dot{\boldsymbol{{\Sigma}}}_t  \boldsymbol{\Psi}_t, \qquad t \in [0,T], \\
		{\boldsymbol{\Psi}_T}&=  \left( \id - \frac 1 {2\lambda} \boldsymbol{\hat K}^* \right)^{-1} A \left( \id - \frac 1 {2\lambda} \boldsymbol{\hat K} \right)^{-1}, 
		\end{align}
		where $\dot{\boldsymbol{{\Sigma}}}_t$ is the strong derivative of $t\mapsto\bold{\Sigma}_t$ induced by the  kernel
		\begin{align}\label{eq:diffkernelsigma}
		\dot{\Sigma}_t(s,u):=-\frac{1}{4\lambda} K(s,t) K(u,t)^\top , \quad a.e.
		\end{align}
	\end{itemize}
\end{lemma}

\begin{proof}
	The proof of (i) follows from \eqref{psi-gamma} and Lemma \ref{lem-bnd-gam}. The proof of (ii) in given in Section \ref{sec-pf-lem-psi}. The proof of (iii) follows the same lines as the proof of Lemma 5.6 in \cite{abi2020markowitz} hence it is omitted. 
\end{proof}

  \subsection{$L^2$--valued BSDE} 
In the following proposition we show that  $\Theta= \{\Theta_{\cdot}(s) : s\in [0,T]\}$ in \eqref{th-sol} is a solution to an $L^2$-valued linear BSDE that involves the operator $\boldsymbol \Psi_t$ and the kernel $\bar \psi_t$ appearing in Lemma~\ref{L:Psi}.  

 \begin{proposition} \label{prop-theta}
For each $s \leq T$, the process $\Theta$ solves the following $L^2$--valued BSDE, 
\be  \label{back-theta} 
\begin{aligned}
d\Theta_t(s) &= \dot{\Theta}_t(s)dt + dN_t(s),  \quad 0 \leq t <s,   \\
\textrm{with } \dot{\Theta}_t(s) &=  2\big(\boldsymbol { {\Psi}}_t \dot{\boldsymbol{\Sigma}}_t \Theta_t  \big)(s) + \bar{\psi}_t(s,t) \E_t[P_t-P_T] e_1,
\end{aligned}
\ee
with the following boundary condition 
\begin{align}\label{theta-bnd}
\Theta_s(s)=-\frac{1}{2\lambda} \left(P_s -   \E_s[P_T]  +\langle  \Theta_s, K_s\rangle_{L^2} \right)e_1, 
\end{align}
where for each $0\leq s \leq T$, $\{N_t(s)\}_{t\geq 0}$ is a suitable square-integrable martingale. 
\end{proposition}
The proof of Proposition is postponed to Section \ref{sec-bsde}.

\subsection{A verification result}

We will use the following lemma that derives the general dynamics for some functionals in $L^2$. 

\begin{lemma}\label{L:dynamics}
	Let $f(t,s)$ and $h(t,s)$ be two $L^2([0,T]^2, \re^2)$ functions that are continuous in $(t,s) \in [0,T]^2$, with partial derivatives with respect to $t$, $\dot f(t,s) := \partial_t f(t,s)$, $\dot h(t,s) := \partial_t h(t,s)$ that are {in $L^2([0,T]^2, \re^2)$}. Define 
	\be \label{f-h} 
	F_t(s):=\mathds{1}_{\{t \leq s\}} f(t,s), \quad \textrm{and}  \quad H_t(s):=\mathds{1}_{\{t \leq s\}} h(t, s). 
	\ee
	Let $\boldsymbol{\Xi}_t := A \id + \bar {\boldsymbol{\Xi}}_t$  where $\bar {\boldsymbol{\Xi}}_t$ a bounded, strongly differentiable and self adjoint integral operator in $L^2$, and $A$ is a $2\times 2$ symmetric matrix.  Then,     
	the derivative of $t\mapsto \langle F_t, \boldsymbol{\Xi}_t H_t \rangle_{L^2}$ is given by
	\begin{align*}
	\frac{d}{dt}\langle F_t, \boldsymbol{\Xi}_t H_t \rangle_{L^2}&= - f^{\top}(t,t)Ah(t,t) -f^{\top} (t,t)\bar{\boldsymbol{\Xi}}_t H_t(t)+\langle \mathds{1}_t \dot f^{} (t, \cdot) , {\boldsymbol{\Xi}}_t H_t\rangle_{L^2} \\ 
&\quad + \langle F_t,  \dot{ {\boldsymbol{\Xi}}}_t H_t \rangle_{L^2}    - h(t,t)^{\top} \bar{\boldsymbol{\Xi}}_t F_t(t)+ \langle F_t,   {\boldsymbol{\Xi}}_t\mathds{1}_t \dot{h}_t \rangle_{L^2}.
	\end{align*}
\end{lemma}

\begin{proof} We first use the following decomposition 
\be \label{re-1} 
\langle F_t, \boldsymbol{\Xi}_t H_t \rangle_{L^2} = \langle F_t, \bar{\boldsymbol{\Xi}}_t H_t \rangle_{L^2} + \langle F_t, (A \id) H_t \rangle_{L^2}.
\ee
Recall that $(A \id) H_t = A H_t$. A direct application of \eqref{f-h},  {\eqref{eq:assumtionG} and a generalized version of Leibnitz's rule (see e.g. Lemma 2.14 in \cite{kalinin})} gives 
\be  \label{re0}
\begin{aligned} 
\frac{d}{dt} \langle F_t, (A \id) H_t \rangle_{L^2} &=  \frac{d}{dt} \int_t^T f^{\top}(t,s)Ah(t,s)ds \\
&= - f^{\top}(t,t)Ah(t,t) + \int_t^T {\dot{f}}^{\top}(t,s)Ah(t,s)ds  \\
&\quad  + \int_t^T f^{\top}(t,s)A\dot{h}(t,s)ds \\
&= - f^{\top}(t,t)Ah(t,t) + \int_{0}^{T}\mathds{1}_{\{s\geq t\}} {\dot{f}}^{\top}(t,s)AH_{t}(s)ds  \\
&\quad  + \int_0^T   F_{t}^{\top}(s)A \mathds{1}_{\{s\geq t\}} \dot{h}(t,s)ds \\
&= - f^{\top}(t,t)Ah(t,t) + \langle {\mathds{1}_{t}\dot{f}}(t,\cdot),A\id H_{t} \rangle_{L^2}  \\
&\quad  + \langle  F_{t} , A  \id\mathds{1}_{ t} \dot{h}(t,\cdot) \rangle_{L^2}, 
\end{aligned} 
\ee 
where we used $(\id H_{t})(s) = H_{t}(s)$ in the last line. 

Using similar arguments, we get 
\be \label{re1} 
	\begin{aligned}
	\frac{d}{dt}\langle F_t, \bar{\boldsymbol{\Xi}}_t H_t \rangle_{L^2} &= \frac{d}{dt} \int_t^T f^{\top} (t,s)\bar{\boldsymbol{\Xi}}_t H_t(s)ds \\ 
&=-f^{\top} (t,t)\bar{\boldsymbol{\Xi}}_t H_t(t)+ \int_t^T \dot f^{\top} (t,s)\bar{\boldsymbol{\Xi}}_t H_t(s)ds \\ 
&\quad + \int_t^T f^{\top} (t,s) \frac{d}{dt} \left( \bar{\boldsymbol{\Xi}}_t H_t(s) \right)ds \\
&=-f^{\top} (t,t)\bar{\boldsymbol{\Xi}}_t H_t(t)+  \langle \mathds{1}_t \dot f  (t, \cdot) , \bar{\boldsymbol{\Xi}}_t H_t\rangle_{L^2} \\ 
&\quad + \int_t^T f^{\top} (t,s) \frac{d}{dt} \left(\bar{\boldsymbol{\Xi}}_t H_t(s)\right)ds. 
 	\end{aligned}
	\ee
Note that 
\be \label{re2} 
	\begin{aligned}
\frac{d}{dt} \left(\bar{\boldsymbol{\Xi}}_t H_t(s)\right) &= \frac{d}{dt} \int_t^T \bar \Xi_t(s,r)h(t,r)dr  \\ 
&= - \bar \Xi_t(s,t)h(t,t) +  \int_t^T  \left(\frac{d}{dt} \bar \Xi_t(s,r)\right) h(t,r)dr  \\ 
&\quad + \int_t^T   \bar \Xi_t(s,r)  \dot{h}(t,r)dr  \\
&=- \bar \Xi_t{(s,t)}h(t,t) +    \dot{ \bar{\boldsymbol{\Xi}}}_t H_t(s)  +   \bar{\boldsymbol{\Xi}}_t\mathds{1}_t \dot h(t, \cdot)(s), 
	\end{aligned}
	\ee
where we used the fact that the kernel of the operator  $\dot{ \bar{\boldsymbol{\Xi}}}_t$ is  $\dot{ \bar{ {\Xi}}}_t$ in the last line. 

From \eqref{re2} and since $\bar{\boldsymbol{\Xi}}_t$ is self adjoint, we get that 
\be \label{re3} 
	\begin{aligned}
&\int_t^T f^{\top} (t,s) \frac{d}{dt} \left(\bar{\boldsymbol{\Xi}}_t H_t(s)\right)ds  \\
&=- \int_t^T f^{\top} (t,s) \bar \Xi_t(s,t)h(t,t) ds 
+\int_t^T f^{\top} (t,s) \dot{ \bar{\boldsymbol{\Xi}}}_t H_t(s) ds \\
 & \quad  + \int_t^T f^{\top} (t,s)  \bar{\boldsymbol{\Xi}}_t\mathds{1}_t \dot h(t, \cdot)(s)ds \\ 
&= - h(t,t)^{\top} \bar{\boldsymbol{\Xi}}_t F_t(t) +   \langle F_t,   \bar{\boldsymbol{\Xi}}_t\mathds{1}_t \dot{h}_t \rangle_{L^2} + \langle F_t,  \dot{ \bar{\boldsymbol{\Xi}}}_t H_t \rangle_{L^2}. 
\end{aligned}
	\ee
From \eqref{re1} and \eqref{re3} it follows that
\be \label{re4} 
	\begin{aligned}
		\frac{d}{dt}\langle F_t, \bar{\boldsymbol{\Xi}}_t H_t \rangle_{L^2} 
&=-f^{\top} (t,t)\bar{\boldsymbol{\Xi}}_t H_t(t)+  \langle \mathds{1}_t \dot f^{ } (t, \cdot) , \bar{\boldsymbol{\Xi}}_t H_t\rangle_{L^2} \\ 
&\quad  - h(t,t)^{\top} \bar{\boldsymbol{\Xi}}_t F_t(t) +   \langle F_t,   \bar{\boldsymbol{\Xi}}_t\mathds{1}_t \dot{h}_t \rangle_{L^2} + \langle F_t,  \dot{ \bar{\boldsymbol{\Xi}}}_t H_t \rangle_{L^2}. 
\end{aligned}
	\ee
Applying  \eqref{re0} and \eqref{re4} to \eqref{re-1} we finally get 
  \begin{align*}
	\frac{d}{dt}\langle F_t,  {\boldsymbol{\Xi}}_t H_t \rangle_{L^2} 
&=-f^{\top} (t,t)\bar{\boldsymbol{\Xi}}_t H_t(t)+  \langle \mathds{1}_t \dot f^{} (t, \cdot) , \bar{\boldsymbol{\Xi}}_t H_t\rangle_{L^2} \\ 
&\quad  - h(t,t)^{\top} \bar{\boldsymbol{\Xi}}_t F_t(t) +   \langle F_t,   \bar{\boldsymbol{\Xi}}_t\mathds{1}_t \dot{h}_t \rangle_{L^2} + \langle F_t,  \dot{ \bar{\boldsymbol{\Xi}}}_t H_t \rangle_{L^2} \\ 
&\quad - f^{\top}(t,t)Ah(t,t) + \langle {\bm{1}_{t}\dot{f}} (t,\cdot),A\id H_{t} \rangle_{L^2}  \\
&\quad  + \langle  F_{t} , A  \id\mathds{1}_{t} \dot{h}(t,\cdot) \rangle_{L^2}  \\
&= - f^{\top}(t,t)Ah(t,t) -f^{\top} (t,t)\bar{\boldsymbol{\Xi}}_t H_t(t)+\langle \mathds{1}_t \dot f^{} (t, \cdot) , {\boldsymbol{\Xi}}_t H_t\rangle_{L^2} \\ 
&\quad + \langle F_t,  \dot{ {\boldsymbol{\Xi}}}_t H_t \rangle_{L^2}    - h(t,t)^{\top} \bar{\boldsymbol{\Xi}}_t F_t(t)+ \langle F_t,   {\boldsymbol{\Xi}}_t\mathds{1}_t \dot{h}_t \rangle_{L^2},
 	\end{align*}
were we used  $\boldsymbol{\Xi}_t := A \id + \bar {\boldsymbol{\Xi}}_t$ and $\dot{\boldsymbol{\Xi}}_t = \dot{ \bar {\boldsymbol{\Xi}}}_t$ in the last line. This completes the proof.

\end{proof}

We will use Lemma \ref{L:dynamics} in order to differentiate the first term on the right hand side of \eqref{eq:candidatevaluefun}. Recall that for $K$ as in \eqref{g0-k-def} we write $K_t(s)= K(s,t)$.
\begin{lemma} \label{lem-der1} 
Let $g^u$ as in \eqref{eq:gu} and $\boldsymbol \Psi $ as in \eqref{def:riccati_operator}  then we have $[0,T] \times \Omega$-a.e.,  
$$
\begin{aligned} 
\frac d {dt}\langle g^u_t, \boldsymbol{\Psi}_t g_t^{u} \rangle_{L^2}  &= - (g^u_t(t))^\top A g^u_t(t) +  \frac{1}{\lambda} g^u_t(t)^\top \left(\boldsymbol{\hat {K}}^*\boldsymbol{\Psi}_t \right)(g^u_t \mathds{1}_t)(t) \\
&\quad + 2 \langle {\boldsymbol{\Psi}}_t K_t ,  g^u_t   \rangle_{L^2} u_t + 2 \langle  g^u_t ,  \boldsymbol{\Psi}_t \dot{\boldsymbol{{\Sigma}}}_t  \boldsymbol{\Psi}_t g^u_t   \rangle_{L^2}. 
\end{aligned} 
$$
\end{lemma}

\begin{proof} 
Define 
\be \label{t-g}
\tilde g^u_t(s):=g_0(s) + \int_0^t K(s,r)u_r dr, \quad 0\leq s, t \leq T,
\ee
where $g_0$ was defined in \eqref{g0-k-def}. {Recall that $h_0$ in \eqref{z-def} is assumed to be continuous, hence $\tilde h_0$ in \eqref{h-tilde} and therefore $g_0$ are also continuous functions}. From \eqref{eq:assumtionG} and an application of Cauchy-Schwarz inequality it follows that $\tilde g^u_t(s)$ is continuous in $(t,s)\in [0,T]^2$, as required by the assumptions of Lemma \ref{L:dynamics}.  

From \eqref{t-g} we get that 
\be \label{t-der}
\frac{d}{dt} \tilde g^u_t(s) = K(s,t)u_t = K_t(s)u_t, \quad  d\P\otimes ds \otimes dt-\rm{a.e.}
\ee
Moreover, from \eqref{eq:gu} we have 
\be \label{g-g-t}
  g^u_t(s) =\mathds{1}_{\{t\leq s\}} \tilde g^u_t(s),  \quad d\P\otimes ds \otimes dt-\rm{a.e.}
\ee
From Lemma~\ref{L:Psi}(i) it follows that, 
\be \label{psi-bar} 
\boldsymbol{\Psi}_t=A\cdot \id + \bar{\boldsymbol{\Psi}}_t, 
\ee
where $\bar{\boldsymbol{\Psi}}_t$ is a self adjoint integral operator. 

From \eqref{t-der}-\eqref{psi-bar} and by a direct application of  Lemma~\ref{L:dynamics} we get a.e. on $[0,T] \times \Omega$-\textrm{a.s.},
\be \label{id34} 
	\begin{aligned}
\frac d {dt}\langle g_t^u, \boldsymbol{\Psi}_t g_t^{u} \rangle_{L^2} &= - (\tilde{g}^u_t(t))^\top A \tilde{g}^u_t(t)  - (\tilde{g}^u_t(t))^\top \bar{\boldsymbol{\Psi}}_t  {g}^u_t(t) +\langle \mathds{1}_t \dot {\tilde{g}}^{u}_t , {\boldsymbol{\Psi}}_t g^u_t\rangle_{L^2} \\
&\quad   + \langle g^u_t,  \dot{ {\boldsymbol{\Psi}}}_t g^u_t \rangle_{L^2} - \tilde g_t^u(t)^{\top} \bar{\boldsymbol{\Psi}}_t g^u_t(t) +\langle g^u_t,   {\boldsymbol{\Psi}}_t\mathds{1}_t \dot{\tilde{g}}^u_t \rangle_{L^2} \\ 
&= - (\tilde{g}^u_t(t))^\top A \tilde{g}^u_t(t)  - (\tilde{g}^u_t(t))^\top ( - A\cdot \id + \boldsymbol{\Psi}_t) {g}^u_t(t) +\langle \mathds{1}_t K_t  , {\boldsymbol{\Psi}}_t g^u_t\rangle_{L^2} u_t\\
&\quad   + \langle g^u_t,  \dot{ {\boldsymbol{\Psi}}}_t g^u_t \rangle_{L^2} - \tilde g_t^u(t)^{\top}( - A\cdot \id + \boldsymbol{\Psi}_t)  g^u_t(t) +\langle g^u_t,   {\boldsymbol{\Psi}}_t\mathds{1}_t K_t  \rangle_{L^2}u_t. 
 \end{aligned}
 \ee
Since $A$ is a symmetric matrix and $\bar{\boldsymbol{\Psi}}_t$ is self-adjoint, it follows from \eqref{psi-bar} that $ {\boldsymbol{\Psi}}_t$ is self-adjoint. Together with \eqref{g-g-t} we get that,  
\be \label{id12} 
\begin{aligned} 
 \langle \mathds{1}_t K_t  , {\boldsymbol{\Psi}}_t g^u_t\rangle_{L^2} & = \langle g^u_t,   {\boldsymbol{\Psi}}_t\mathds{1}_t K_t  \rangle_{L^2}  \\
&= \langle  {\boldsymbol{\Psi}}_t K_t ,  g^u_t   \rangle_{L^2}. 
\end{aligned} 
\ee
From \eqref{g-g-t} we also have $g^u_t(t) = \tilde g^u_t(t)$. Using this and \eqref{id12}, we can gather similar terms in \eqref{id34} and get, 
	\begin{align*}
\frac d {dt}\langle g_t^u, \boldsymbol{\Psi}_t g_t^{u} \rangle_{L^2}  &= (g^u_t(t))^\top A g^u_t(t) -  2 g^u_t(t)^\top({\bar{{\boldsymbol{\Psi}}}_t  } g^u_t)(t)  \\
&\quad + 2 \langle  {\boldsymbol{\Psi}}_t K_t ,  g^u_t   \rangle_{L^2} u_t +  \langle  g^u_t , \dot{\boldsymbol{\Psi}}_t g^u_t   \rangle_{L^2}, \quad   d\P  \otimes dt-\rm{a.e.}
\end{align*}
Together with Lemma~\ref{L:Psi}(ii) and (iii) it follows that  
	\begin{align*}
\frac d {dt}\langle g_t^u, \boldsymbol{\Psi}_t g_t^{u} \rangle_{L^2}  
&= - (g^u_t(t))^\top A g^u_t(t) +  \frac{1}{\lambda} g^u_t(t)^\top \left(\boldsymbol{\hat {K}}^*\boldsymbol{\Psi}_t \right)(g^u_t \mathds{1}_t)(t) \\
&\quad + 2 \langle {\boldsymbol{\Psi}}_t K_t ,  g^u_t   \rangle_{L^2} u_t + 2 \langle  g^u_t ,  \boldsymbol{\Psi}_t \dot{\boldsymbol{{\Sigma}}}_t  \boldsymbol{\Psi}_t g^u_t   \rangle_{L^2}, \quad  d\P  \otimes dt-\rm{a.e.},
\end{align*}
which completes the proof. 
\end{proof} 

We now use Lemma \ref{L:dynamics} to differentiate the second term on the right hand side of \eqref{eq:candidatevaluefun}.  
\begin{lemma} \label{lem-der2} 
Let $g^u$ as in \eqref{eq:gu} and $\Theta_t$ as in \eqref{back-theta} . Then we have a.e. on $[0,T]\times \Omega$,
\begin{align*}d\langle {\Theta}_t , g_t^u \rangle_{L^2}  
& =  \left(2 \langle \boldsymbol{\Psi}_t \dot{\boldsymbol{\Sigma}}_t \Theta_t, g_t^u \rangle_{L^2} {-}\frac 1 {2\lambda} (P_t-M_t) \langle \boldsymbol{\Psi}_tK_t, g_t^u  \rangle_{L^2}\right) dt \\
&\quad + \left(   \langle {\Theta}_t , K_t \rangle_{L^2} u_t +  \frac{1}{2\lambda} \left(P_t - M_t  +\langle  \Theta_t, K_t\rangle_{L^2} \right)Y^u_t \right) dt + \langle dN_t, g_t^u \rangle_{L^2} .
 \end{align*}
\end{lemma} 

\begin{proof}  
First note that Lemma~\ref{L:Psi}(i) and (ii)  can be applied to $g_t^u$ instead of $f$, as time dependence will not change the result. Together with \eqref{g0-k-def}, \eqref{eq:explicithatK} and \eqref{eq:gu} we get for every $0\leq t\leq T$,
\be \label{n-id} 
\begin{aligned}
\langle  \bar{\psi}_t(\cdot,t)e_1, g_t^u\rangle_{L^2} &= e_1^\top  \int_t^T \bar{\psi}_t(t,s) g_t^u(s) ds\\
 & = e_1^\top (\boldsymbol{\Psi}_t-A\id)(g^u_t\mathds{1}_t)(t)\\
&= \frac{1}{2\lambda} e_1^\top (\hat{\boldsymbol{K}}^*\boldsymbol{\Psi}_t)(g^u_t \mathds{1}_t)(t) \\
&= - \frac{1}{2\lambda} e_1^\top e_1 \langle \boldsymbol{\Psi}_t K_t, g^u_t\mathds{1}_t \rangle_{L^2} \\
&= - \frac{1}{2\lambda} \langle \boldsymbol{\Psi}_t K_t, g^u_t \mathds{1}_t \rangle_{L^2}. 
\end{aligned}
 \ee
Using Leibnitz rule and \eqref{g-g-t}, we get a.e. on $[0,T]\times \Omega$,
\be \label{nn1} 
 \begin{aligned}
d\langle {\Theta}_t , g_t^u \rangle_{L^2} &= d\int_{t}^{T} \Theta^\top_t(s) \tilde g_t^u(s)ds \\ 
&=  - \Theta_t(t)^\top g_t(t)dt  +  \int_{t}^{T} d\big(\Theta_t^\top(s) {\tilde g}_t^u(s)\big) ds. 
 \end{aligned}
 \ee
From It\^{o} product rule, \eqref{back-theta}  and \eqref{t-der} we have a.e. on $[0,T]\times \Omega$
 \be\label{nn2} 
 \begin{aligned}
 d \big(\Theta^\top_t(s) \tilde g_t^u(s)\big) 
&=     \tilde g_t^u(s)  d\Theta^\top_t(s) dt  +    \Theta_t^\top(s) \dot{\tilde g}_t^u(s) dt\\  
&=   (\dot{\Theta}^\top_t(s)dt   +dN^\top_t )\tilde g_t^u(s) +  \Theta^\top_t(s) \dot{\tilde g}_t^u(s)dt.   
 \end{aligned}
 \ee
From \eqref{t-der}, \eqref{nn1} and \eqref{nn2} we get, 
 \begin{align*}
d\langle {\Theta}_t , g_t^u \rangle_{L^2}  
&= \left(\langle \dot{\Theta}_t , g_t^u \rangle_{L^2}  +   \langle {\Theta}_t , K_t \rangle_{L^2} u_t - \Theta_t(t)^\top g_t(t) \right) dt + \langle dN_t, g_t^u \rangle_{L^2}. 
 \end{align*}
Together with \eqref{g-def}, \eqref{back-theta}, \eqref{theta-bnd} and \eqref{n-id} we get a.e. on $[0,T]\times \Omega$,
  \begin{align*}
d\langle {\Theta}_t , g_t^u \rangle_{L^2} & =  \left(2 \langle \boldsymbol{\Psi}_t \dot{\boldsymbol{\Sigma}}_t \Theta_t, g_t^u \rangle_{L^2} { -}\frac 1 {2\lambda} (P_t-M_t) \langle \boldsymbol{\Psi}_tK_t, g_t^u  \rangle_{L^2}\right) dt \\
&\quad + \left(   \langle {\Theta}_t , K_t \rangle_{L^2} u_t +  \frac{1}{2\lambda} \left(P_t - M_t  +\langle  \Theta_t, K_t\rangle_{L^2} \right)Y^u_t \right) dt + \langle dN_t, g_t^u \rangle_{L^2}.
 \end{align*}
\end{proof} 

Now we  are ready to prove Theorem \ref{eq:MainTheorem}.

\begin{proof}[Proof of Theorem \ref{eq:MainTheorem}] 
{The uniqueness of the optimal trading speed $u^*$ follows from the strict concavity of the cost functional $u \to J(u)$ for $u\in \mathcal A$. The proof of the concavity of  $J$ in \eqref{def:objective} follows the same lines as the proof of Theorem 2.3 in \cite{Lehalle-Neum18}. See also Lemma 10.1 \cite{N-Voss-23} which incorporates the temporary price impact term in the cost functional.  We therefore omit the details.}

Recall that the proposed value process $V^u_t$ was defined in \eqref{eq:candidatevaluefun} and the performance functional $J$ was defined in \eqref{def:objective}.  For any  admissible $u$ as in \eqref{def:admissset} we set 
	 \be  \label{j-t-func} 
	 J_t(u):=\E\left[\int_{t}^{T}  \left( (P_s-Y^u_s) u_s - \lambda u_s^2 - \phi (Q^u_s)^2 \right)ds +  P_T Q^u_T - \varrho q^2 \Mid \mathcal F_t\right],
	 \ee
	 and we define the  process 
	\be \label{m-mart} 
	M^{u}_t := 	\int_0^t \left( (P_s-Y^u_s) u_s - \lambda u_s^2 - \phi (Q^u_s)^2 \right)ds + V_t^{u} + \lambda \int_0^t  (u_s - {\mathcal T_s(u)})^2  d s,
	\ee
	where  
	\be \label{T-def}
	{\mathcal T_t(u)}:= \frac 1 {2\lambda} \big(   \E[ (P_t - P_T) \mid \mathcal F_t]  - Y^u_t + \langle \Theta_t , K_t \rangle_{L^2} +  \langle \boldsymbol{\Psi}_t K_t  , g_t^u \rangle_{L^2} \big).
\ee
Proposition \ref{prop-mart} below implies that 
\be \label{mart} 
\E[M^{u}_{T}| \mathcal F_t]=M^{u}_{t}, \quad \textrm{for all } 0\leq t \leq T, \ \P-\textrm{a.s.} 
\ee

From \eqref{eq:gu} it follows that $g^u_T(s)=0$ on $[0,T)$. From \eqref{eq:chi} we have $\chi_T = -2\varrho q^2$. Using both terminal conditions on \eqref{eq:candidatevaluefun} give,  
\be \label{terminal:candidatevaluefun}
\begin{aligned}V_T^{u}&=\frac{1}{ 2}\langle g_T^{u}, \boldsymbol{\Psi}_T g_T^{u} \rangle_{L^2} + \langle \Theta_T, g_T^{u} \rangle_{L^2} +  \mathbb E[P_T \mid \mathcal F_T] Q_T^{u} + \frac{1}{ 2} \chi_T \\
&= P_T  Q_T^{u}-  \varrho q^2.
\end{aligned}  
\ee
By using \eqref{mart} and \eqref{terminal:candidatevaluefun} on \eqref{j-t-func}  and \eqref{m-mart} it follows that  
	\begin{align}\label{eq:tempalpha}
 V_t^u - J_t(u) & = \lambda \E\left[ \int_t^{T}  (u_s - {\mathcal T_s(u)})^2 d s  \Mid \mathcal {F}_t \right],  \ \textrm{for all } 0\leq t\leq T, \ \P-\rm{a.s.}
	\end{align}
	Since $\lambda >0$, the right hand side of \eqref{eq:tempalpha} is always nonnegative and it vanishes for $u=u^*$, where $u^*(\cdot) = \mathcal T_{\cdot}(u^*)$ is given by \eqref{eq:optimalcontrol_monotone}.  
	
Fix now $t \leq T$ and recall that $\mathcal A_t(u^*)$ was defined in \eqref{admis-u}. We observe  that $V_t^{u^*}$ $=$ $V_t^{u'}$ for all 
	$u'$ $\in$ $\mathcal A_t(u^*)$. We then deduce from \eqref{eq:tempalpha}  that
	\begin{align*}
	V_t^{u^*} & = J_t(u^*) \; = \; \operatorname*{ess~sup}_{u'\in\mathcal A_t(u^*)} J_t(u'), 
	\end{align*}
	which is equivalent to \eqref{eq:optimalvalue_monotone}. Finally, we note that the admissibility of $u^*$ follows from the explicit solution derived in Proposition \ref{prop-exp-u} which is equivalent to \eqref{eq:optimalcontrol_monotone} and from the bounds derived in Lemma \ref{bnd-coef}. This proves that $u^*$ is an optimal control. 
\end{proof}

The following proposition has been used in the proof above. 
\begin{proposition} \label{prop-mart} 
For any $u$ $\in$ $\mathcal A$, $M^{u}$ is a martingale with respect to $(\mathcal F_t)_{0 \leq t\leq T}$. 
\end{proposition}

Before proving  Proposition~\ref{prop-mart}, we collect in the next lemma,  some useful identities. Recall that the notation $\mathds{1}_t$ was introduced in \eqref{ind-def}. 
\begin{lemma}  \label{lem-id} 
The following identities hold: 
\begin{itemize} 
\item[\textbf{(i)}] 
 \be
     4 \langle \boldsymbol{\Psi}_t \dot{\boldsymbol{\Sigma}}_t \Theta_t, g_t^u \rangle_{L^2}     = - \frac{1}{\lambda} \langle \boldsymbol{\Psi}_t K_t, g_t^u \rangle_{L^2} \langle K_t,  \Theta_t \rangle_{L^2}, \quad \textrm{for all } 0\leq t\leq T, \ \P-\rm{a.s.} 
 \ee
\item[\textbf{(ii)}] 
 $$
  \left(\boldsymbol{\hat {K}}^*\boldsymbol{\Psi}_t \right)(g^u_t \mathds{1}_t)(t) =- \langle \boldsymbol{\Psi}_tK_t, g_t \rangle_{L^2} e_1. 
  $$
\item[\textbf{(iii)}] 
  \be \label{de2} 
 2  \langle   \boldsymbol{\Psi}_t\dot{\boldsymbol{\Sigma}}_t \boldsymbol{\Psi}_t g_t^u, g_t^u \rangle_{L^2}  =  -\frac{1}{2\lambda}      \langle \boldsymbol{\Psi}_t K_t  , g_t^u \rangle_{L^2}^2.
\ee
\end{itemize} 
\end{lemma} 
The proof of Lemma \ref{lem-id} is postponed to Section \ref{sec-pf-id}.

\begin{proof} [Proof of Proposition \ref{prop-mart}] 	
Recall that $M_{t} = \E[P_T | \mathcal F_t]$ and define  
\be \label{beta} 
\beta_t= P_t - M_t +  \langle \Theta_t , K_t \rangle_{L^2}. 
\ee
Then from \eqref{m1-mart}, \eqref{m-mart}, \eqref{T-def} and \eqref{beta} we have  
\be \label{m-dif} 
\begin{aligned}
dM^{u}_t &=   \left((P_t-Y^u_t) u_t - \lambda u_t^2 - \phi (Q^u_t)^2     + \lambda (u_t - {\mathcal T_t(u)})^2 \right)dt + dV_t^{u}\\
&=   \big((P_t-Y^u_t) u_t - \phi (Q^u_t)^2    -  2\lambda {\mathcal T_t(u)}u_t  + \lambda {\mathcal T_t(u)}^2\big)dt  + dV_t^{u}\\
&= - \phi (Q^u_t)^2 dt  -   \left( - M_t +  \langle \Theta_t , K_t \rangle_{L^2} +    \langle \boldsymbol{\Psi}_t K_t  , g_t^u \rangle_{L^2} \right) u_t \big)dt  \\
&\quad + \frac 1 {4\lambda} \left( \beta_t - Y^u_t +    \langle \boldsymbol{\Psi}_t K_t  , g_t^u \rangle_{L^2} \right)^2dt + dV_t^{u} \\
&= \left(\frac {1}{2}g^u_t(t)^\top A g^u_t(t)   -  \left( - M_t +  \langle \Theta_t , K_t \rangle_{L^2} +    \langle \boldsymbol{\Psi}_t K_t  , g_t^u \rangle_{L^2} \right) u_t \right)dt\\
&\quad  + \frac {1}{4 \lambda} \left( \beta_t^2 + 2\beta_t  \langle \boldsymbol{\Psi}_t K_t  , g_t^u \rangle_{L^2} - 2\beta_t Y_t^u - 2 Y_t^u  \langle \boldsymbol{\Psi}_t K_t  , g_t^u \rangle_{L^2} +  \langle \boldsymbol{\Psi}_t K_t  , g_t^u \rangle_{L^2}^2   \right)dt \\
&\quad + dV_t^{u},  \quad \textrm{for all } 0\leq t\leq T, \ \P-\rm{a.s.}
\end{aligned}
\ee
where we have used the identity   
$$g^u_t(t)^\top A g^u_t(t) = \frac {1}{2\lambda}(Y^u_t)^{2} - 2\phi (Q^u_t)^2, 
$$
  which follows from \eqref{g-def} and \eqref{a-mat}, in the last equality.  

An application of \eqref{eq:BSDEchi} and Lemmas~\ref{lem-der1} and~\ref{lem-der2} to \eqref{eq:candidatevaluefun} yields:
\be \label{v-dif} 
\begin{aligned}
	dV^{u}_t &= {\frac{1}{2}}\bigg(  \left(-g^u_t(t)^\top A g^u_t(t) +  \frac{ {1}}{\lambda} g^u_t(t)^\top  \left(\boldsymbol{\hat {K}}^*\boldsymbol{\Psi}_t \right)(g^u_t \mathds{1}_t)(t) \right)dt\\
	&\quad +\left( 2 \langle {\boldsymbol{\Psi}}_t K_t ,  g^u_t   \rangle_{L^2} u_t + 2 \langle  g^u_t ,  \boldsymbol{\Psi}_t \dot{\boldsymbol{{\Sigma}}}_t  \boldsymbol{\Psi}_t g^u_t   \rangle_{L^2} \right)dt\\
	&\quad+ \left(4 \langle \boldsymbol{\Psi}_t \dot{\boldsymbol{\Sigma}}_t \Theta_t, g_t^u \rangle_{L^2} - \frac {1} {\lambda} (P_t-M_t) \langle \boldsymbol{\Psi}_tK_t, g_t^u  \rangle_{L^2}\right) dt \\
&\quad + \left(  2 \langle {\Theta}_t , K_t \rangle_{L^2} u_t +  \frac{ {1}}{\lambda} \left(P_t - M_t  +\langle  \Theta_t, K_t\rangle_{L^2} \right)Y^u_t \right) dt \\
	&\quad + \left( -2M_t u_t +  \dot \chi_t \right) dt + 2Q^u_t dM_t  + d\wt{M}_t +2 \langle dN_t, g^u_t \rangle_{L^2} \bigg), \quad dt\otimes d\P-\rm{a.e}. 
	\end{aligned}
	\ee
	Note that from \eqref{eq:BSDEchi} and \eqref{beta} we have 
\be \label{xi-beta}
 \dot{\chi}_t = -\frac 1 {2\lambda} \beta_t^2.
\ee
Using \eqref{xi-beta} and Lemma~\ref{lem-id}, we can rewrite \eqref{v-dif} as 
\be  \label{v-dif2} 
\begin{aligned}
	dV^{u}_t &= \frac{1}{2}\bigg( \left(- g^u_t(t)^\top A g^u_t(t) {-} \frac{1}{\lambda} g^u_t(t)^\top   \langle \boldsymbol{\Psi}K_t, g_t \rangle_{L^2} e_1 \right)dt\\
	&\quad + \left(2 \langle {\boldsymbol{\Psi}}_t K_t ,  g^u_t   \rangle_{L^2} u_t   -\frac{1}{2\lambda}      \langle \boldsymbol{\Psi}_t K_t  , g_t^u \rangle_{L^2}^2 \right)dt \\
	&\quad- \frac {1} { \lambda} \left( \langle \boldsymbol{\Psi}_t K_t, g_t^u \rangle_{L^2} \langle K_t,  \Theta_t \rangle_{L^2}+  (P_t-M_t) \langle \boldsymbol{\Psi}_tK_t, g_t^u  \rangle_{L^2}\right) dt \\
&\quad + \left(  2 \langle {\Theta}_t , K_t \rangle_{L^2} u_t +  \frac{1}{\lambda} \left(P_t - M_t  +\langle  \Theta_t, K_t\rangle_{L^2} \right)Y^u_t \right) dt \\
	&\quad - \left( 2M_t u_t   +\frac 1 {2\lambda} \beta_t^2 \right) dt + 2Q^u_t dM_t+ d\wt M_t + 2 \langle dN_t, g^u_t \rangle_{L^2} \bigg), \quad dt\otimes d\P-\rm{a.e}. 
	\end{aligned}
\ee
Plugging in \eqref{v-dif2} to \eqref{m-dif} and using \eqref{beta} we get  
\bd
			\begin{aligned}
		dM^{u}_t &=     \left(\frac{1}{{2}\lambda} g^u_t(t)^\top e_1   \langle \boldsymbol{\Psi}K_t, g_t \rangle_{L^2}  -\frac{1}{{2}\lambda} Y_t^u  \langle \boldsymbol{\Psi}_t K_t  , g_t^u \rangle_{L^2} \right)dt   \\
			&\quad  +   Q^u_tdM_t+  \frac{1}{2}d\wt M_t +{ \langle dN_t, g^u_t \rangle_{L^2} }, \quad dt\otimes d\P-\rm{a.e}. 
\end{aligned}
\ed
From \eqref{g-def} we have $g^u_t(t)^\top e_1 = Y_t^u$, so we get 
\be  \label{m-u-def}
dM_t^u =   Q^u_t dM_t    + \frac{1}{2} d\wt M_t +{ \langle dN_t, g^u_t \rangle_{L^2}} , \quad dt\otimes d\P-\rm{a.e}. 
\ee
Note that from \eqref{def:Q},\eqref{eq:gu},\eqref{m1-mart},\eqref{th-sol},\eqref{eq:BSDEchi},\eqref{eq:zxprop},\eqref{eq:candidatevaluefun} and the explicit presentation of $(N_t)_{t\in [0,T]}$ which is detailed in \eqref{eq:dN}, one can show that the right-hand side of \eqref{m-u-def} admits a càdlàg modification. Similar observation is made for the left-hand side of \eqref{m-u-def} using \eqref{m-mart} and \eqref{eq:candidatevaluefun}. Hence it follows that modifications of both sides of \eqref{m-u-def} hold in equality for all $t\in [0,T]$, $\P$-a.s. This shows that $M^u$ is a local martingale. 
To argue true martingality, we bound every term in the right-hand side of \eqref{m-mart} in order to get 
\be \label{hj0} 
\E\Big[\sup_{t\in [0,T]} |M_t^u|\Big]<\infty, \quad \textrm{ for any } u \in \mathcal A,
\ee
where we recall that the set of admissible controls $\mathcal A$ was defined in \eqref{def:admissset}. 
An application of Jensen's inequality and Cauchy–Schwarz's inequality gives,
\be \label{hj1} 
\begin{aligned} 
\E \left[ \sup_{t\in [0,T]} \left(\int_0^t P_s u_s  ds \right)^2 \right] &\leq C \E \left[ \left( \int_0^T |P_s| |u_s|  ds \right)^2 \right] \\
&\leq C \E \left[  \int_0^T P^2_sds \right] \E \left[  \int_0^T u^2_sds \right] \\ 
&<\infty,   \quad \textrm{for all } u \in \mathcal A,
\end{aligned} 
\ee
where we have used \eqref{ass:P} and \eqref{def:admissset} in the last inequality. 

Additional applications of Jensen's inequality and Cauchy–Schwarz's inequality give, 
\be \label{hj2} 
\begin{aligned}
&\E \left[ \sup_{t\in [0,T]} \left( \int_0^t \left(\int_0^s G(s,r)u_r dr \right)  u_s  ds  \right)^2 \right] \\
&\leq  C \E \left[   \int_0^T \left(\int_0^s G(s,r)u_r dr \right)^2 ds\right]  \E \left[  \int_0^Tu^2_s ds \right] \\
&\leq  C(T) \E \left[   \int_0^T \left(\int_0^s G(s,r)u_r dr \right)^2 ds \right]\\ 
&\leq  C(T) \left( \sup_{s\in [0,T]}\int_0^T G^2(s,r)dr \right) \E \left[ \int_0^T u^2_r dr \right]  \\ 
&<\infty,  \quad \textrm{for all } u \in \mathcal A,
\end{aligned} 
\ee
where we have used \eqref{def:admissset} and \eqref{eq:assumtionG} in the last inequality.

From \eqref{def:Q}, \eqref{def:admissset}, \eqref{z-def}, \eqref{y-z} and \eqref{hj2} we get,  
\be \label{tr-4} 
\begin{aligned}
\E \left[ \sup_{t\in [0,T]} \left|\int_0^t \left( Y^u_su_s - \lambda u_s^2 - \phi (Q^u_s)^2 \right)ds  \right| \right] <\infty,  \quad \textrm{for all } u \in \mathcal A. 
\end{aligned} 
\ee

Recalling \eqref{m-mart}, our next step is to prove that 
\be\label{tr-3}
\E \left[\sup_{t\in [0,T]} |V_t^{u} |  \right] < \infty, \quad \textrm{for any } u \in \mathcal A, 
\ee
where $V_{\cdot}^{u}$ was defined in  \eqref{eq:candidatevaluefun}. 

The boundedness properties of $\Theta= \{\Theta_{t}(s) : t\in [0,s],\, s\in [0,T]\}$ defined in \eqref{th-sol} are detailed in Lemma \ref{lem-bnd-theta} below and give,   
$$
\E\left[  \sup_{t \in [0,T]} \int_t^T \mathbb  |\Theta_t(s)|^2 ds\right]   <\infty. 
$$
Using this bound together with Jensen and Cauchy–Schwarz inequalities and \eqref{eq:gu} we get, 
\be\label{tr-2} 
\begin{aligned}
&\E\left[  \sup_{t \in [0,T]} ( \langle \Theta_t, g_t^{u} \rangle_{L^2} )^2\right]  \\
&= \E\left[ \sup_{t \in [0,T]}  \left| \int_{t}^T \Theta_t(s) g_t^{u}(s) ds  \right| \right] \\
&\leq  C(T) \E\left[  \sup_{t \in [0,T]} \int_t^T    |\Theta_t(s)|^2 ds\right]  \E\left[  \sup_{t \in [0,T]} \int_t^T  \left| \tilde h_{0}(s) + \int_{0}^{t}\tilde G(s,r)u_{r}dr \right|^2 ds\right] \\
&\leq  C_1(T)\left ( C_2(T)+  \E\left[  \sup_{t \in [0,T]} \int_t^T  \left| \int_{0}^{t}\tilde G(s,r)u_{r}dr \right|^2 ds\right] \right) \\
&< \infty, \quad \textrm{for all } u \in \mathcal A,
\end{aligned} 
\ee
where the last inequality follows from \eqref{h-tilde} and similar steps as in \eqref{hj2}. 

Using similar steps as in \eqref{hj1} we get, 
\be\label{tr-1} 
\E \left[ \sup_{t\in [0,T]} \big| \mathbb E[P_T \mid \mathcal F_t] Q_t^{u}  \big| \right] <\infty. 
\ee
The following bound on $\boldsymbol{\Gamma}^{-1}_t$ from \eqref{op-gam-inv} will be proved later in Lemma \ref{lem-bnd-gam}
\be \label{tr1} 
\sup_{t\leq T} \|\boldsymbol{\Gamma}^{-1}_t \|_{\rm{op}}   < \infty.
 \ee
 Repeating the same steps as in \eqref{hj2}, using \eqref{eq:gu}, \eqref{def:admissset} and \eqref{eq:assumtionG} we get, 
 \be\label{tr2} 
\E\Big[ \sup_{t\in [0,T]}\|g_t^u\|^2_{L^2} \Big]< \infty, \quad \textrm{for all } u \in \mathcal A. 
 \ee
From \eqref{op-norm}, \eqref{tr1} and \eqref{tr2} it follows that 
\be\label{tr3}
\begin{aligned}
\E \left[\sup_{t\in [0,T]} \big|\langle g_t^{u}, \boldsymbol {\Gamma}_t^{-1} g_t^{u} \rangle_{L^2}  \big| \right] &\leq \E \left[\sup_{t\in [0,T]} \|g_t^u\|_{L^2} \| \boldsymbol {\Gamma}_t^{-1} g_t^{u} \|_{L^2} \right] \\ 
&\leq C\E \left[\sup_{t\in [0,T]} \|g_t^u\|^2_{L^2}  \right] \\ 
&< \infty, \quad \textrm{for all } u \in \mathcal A. 
\end{aligned} 
\ee
Using similar steps as \eqref{hj2}--\eqref{tr3} in on \eqref{eq:BSDEchi} and \eqref{eq:zxprop} we can derive the following bound,
\be\label{tr4}
\mathbb E\left[ \sup_{t\in [0,T]}|\chi_t|  \right]<\infty.
\ee
From \eqref{eq:candidatevaluefun}, \eqref{tr-2}, \eqref{tr-1}, \eqref{tr3} and \eqref{tr4}, \eqref{tr-3} follows. 

Recall that $\mathcal T_{\cdot}$ was defined in \eqref{T-def}. Next we will show that  
\be \label{tr4.5} 
\E\left[ \int_0^T (u_s - {\mathcal T_s(u)})^2  d s \right] <\infty,   \quad \textrm{for all } u \in \mathcal A.
\ee
Note that 
\be\label{tr5}
\begin{aligned}
\E\left[ \int_0^T (u_s - {\mathcal T_s(u)})^2  d s \right] &\leq 2\E\left[ \int_0^T u_s^2  d s \right] +2\E\left[ \int_0^T  {\mathcal T_s(u)}^2  d s \right], 
\end{aligned} 
\ee
hence by \eqref{def:admissset} it suffices to bound the second term in the right hand side of \eqref{tr5}. 

Considering the terms on the right-hand side of \eqref{T-def}, it follows that it is enough to derive the following bound  in order to establish \eqref{tr4.5} as the rest of the terms can be bounded using similar arguments as in \eqref{tr-4}, \eqref{tr-2} and \eqref{tr-1}. From \eqref{eq:assumtionG}, \eqref{g0-k-def} , \eqref{eq:gu}, \eqref{h-tilde}, \eqref{K-t},  and Cauchy–Schwarz inequality we get, 
\be\label{tr5}
\begin{aligned}
\int_0^T \E\big(|\langle \boldsymbol{\Psi}_t K_t  , g_t^u \rangle_{L^2}| \big)^2 dt  &\leq \int_0^T\| \boldsymbol{\Psi}_t K_t \|^2_{L^2} \E\big[\| g_t^u \|^2_{L^2} \big]dt \\
&\leq C  \int_0^T     \|K_t \|^2_{L^2}  \left(\int_{0}^T \int_{0}^T |\psi_t  (s,r) |^2ds dr  \right)dt   \\
&<\infty. 
\end{aligned} 
\ee
where we have also used \eqref{tr2} in the second inequality and Lemma \eqref{L:Psi}(i) in the last inequality. It follows that \eqref{tr4.5} is satisfied.  
From \eqref{m-mart}, \eqref{tr-3}, \eqref{tr-4} and \eqref{tr4.5} we get \eqref{hj0}, hence $M^u$ is a true martingale.   
\end{proof} 

\section{Proof of Proposition \ref{prop-exp-u}}  \label{sec-prop-expl} 
\begin{proof}[Proof of Proposition \ref{prop-exp-u}]
Recalling \eqref{def:Q} we note that $g^u$ defined in  \eqref{eq:gu} can be re-written in the form   
\begin{align*}
    g_t^u(s) = \mathds 1_{\{s\geq t\}} \left(  \tilde h_0(s), q\right)^\top + \int_0^t \mathds 1_{\{s\geq t\}} \left(\tilde G(s,r), -1 \right)^\top u_r dr, 
\end{align*}
so that an application of Fubini's theorem leads to
\begin{align*}
  \langle \boldsymbol {\Gamma}_t^{-1} K_t  , g_t^{u} \rangle_{L^2} = \langle \boldsymbol {\Gamma}_t^{-1} K_t  , \mathds 1_t \left( \tilde h_0, q \right)^\top  \rangle_{L^2} + \int_0^t     \langle \boldsymbol {\Gamma}_t^{-1} K_t  , \mathds 1_t \left( \tilde G(\cdot, r), - 1  \right)^\top \rangle_{L^2} u_r dr,
\end{align*}
which, combined with \eqref{y-h-tilde}, yields that we can rewrite \eqref{eq:optimalcontrol_monotone} as \eqref{volt-u}. Note that \eqref{volt-u} is a linear Volterra equation which admits a solution for any fixed $\omega \in \Omega$, whenever $a(\omega)\in L^2([0,T],\mathbb R)$ and $B$ satisfies $$\sup_{t\leq T} \int_0^T B(t,s)^2 ds < \infty. $$  Indeed, the solution is given in terms of the resolvent  $R^B$ of B, see \eqref{eq:resolventeqkernel} below,  which exists by virtue of Corollary 9.3.16 in \cite{gripenberg1990volterra} and satisfies
\begin{align*}
\int_0^T \int_0^T |R^B(t,s)| dt ds < \infty.
\end{align*}
In this case, the solution $u^*$ is given by
 \begin{align*}
    u_t^* = a_t + \int_0^t R^B(t,s) a_s ds. 
 \end{align*}
 Note that $R^B$ is the kernel of the operator given by 
 $$
\boldsymbol R^B =  (\id - \boldsymbol B)^{-1}. 
$$
One would still need to check that $u^* \in \mathcal A$ defined as in \eqref{def:admissset}. This follows from the Lemma \ref{bnd-coef} below.  
\end{proof} 
\blue{
\begin{lemma} \label{bnd-coef} 
Assume $\lambda >0$ and $ \phi,\varrho \geq 0$. Then, the following hold:  {
\begin{itemize} 
\item[\bf{(i)}]  $\mathbb E\left[  \sup_{t\leq T} a_t^2 \right] <\infty$, 
\item[\bf{(ii)}] $\sup_{t\leq T} \int_0^T B(t,s)^2 ds < \infty, $ 
\item[\bf{(iii)}] {$\mathbb E\left[ \sup_{t\leq T} (u_t^*)^2  \right] < \infty$.} 
\end{itemize} }
\end{lemma} }
The rest of this section is dedicated to the proof of Lemma \ref{bnd-coef}. In order to prove this lemma we will need some auxiliary results. 
Recall that the operator norm was defined in \eqref{op-norm}. 
\begin{lemma} \label{lem-d-bnd} Assume that  $\phi, \varrho \geq 0$ and $\lambda>0$. Then 
$$\sup_{t\leq T} \|\boldsymbol{D}^{-1}_t \|_{\rm{op}}   < \infty. $$
\end{lemma} 
\begin{proof} 
Choose $\eps \in (0,\lambda)$. In the proof of Lemma \ref{lemmma-inv-op} we have shown that $(\boldsymbol{\tilde G}_t + \boldsymbol{\tilde G}^*_t)$ and  $  \boldsymbol{1}^*_t \boldsymbol{1}_t$ are non-negative definite for any $0\leq t \leq T$. Together with \eqref{eq:assumtionG} and \eqref{h-tilde} it follows that for any $0\leq t\leq T$, the operator 
\be \label{s-op} 
\boldsymbol S_t := 2(\lambda -\eps) \id  + (\boldsymbol{\tilde G}_t + \boldsymbol{\tilde G}^*_t) + 2\phi  \boldsymbol{1}^*_t \boldsymbol{1}_t 
\ee
is positive definite, invertible, self-adjoint and compact with respect to the space of bounded operators {on $L^{2}([0,T])$} equipped with the operator norm given in \eqref{op-norm}. From Theorem 4.15 in \cite{Porter1990} it follows that  $\boldsymbol S_t$ admits a spectral decomposition in terms of a sequence of positive eigenvalues $(\mu_{t,n})_{n=1}^{\infty}$ and an orthonormal sequence of eigenvectors $(\varphi_{t,n})_{n=1}^{\infty}$ in $L^{2}([0,T])$ such that it holds that
$$
\boldsymbol S_t= \sum_{k}\mu_{t,k} \langle \varphi_{t,k}, \cdot \rangle_{L^2} \varphi_{t,k}. 
$$
By application of Cauchy Schwarz and the fact that $\boldsymbol S_t$ is self-adjoint we get  
\begin{align*}
\sup_{t\leq T}\sum_{k}\mu_{t,k}^2 & \leq  C\left((\lambda -\eps)^2 + \sup_{t\leq T} \int_0^T \left(({\tilde G}_t + {\tilde G}^*_t) + 2\phi  \mathds{1}^*_t \mathds{1}_t \right)^2 {(s, s)  ds} \right)\\
&<\infty, 
\end{align*}
where the second inequality follows from \eqref{eq:assumtionG} and \eqref{h-tilde}. 
From \eqref{eq:schur} and \eqref{s-op} it follows that we can rewrite  $\boldsymbol D_t= \boldsymbol S_t +\eps \id$ as follows, 
$$ \boldsymbol D_t = \sum_{k} \left(2\eps + \mu_{t,k}\right) \langle  \varphi_{t,k}, \cdot \rangle_{L^2}  \varphi_{t,k}. $$
We can therefore represent $\boldsymbol{D}_t^{-1}$ as follows, 
$$ \boldsymbol{D}_t^{-1} = \sum_{k}  \frac{1} {\left(2\eps + \mu_{t,k}\right)}\langle \varphi_{t,k}, \cdot \rangle_{L^2} \varphi_{t,k}.$$
Since $\eps>0$ and $\mu_{t,k} \geq 0$, for all $t\in [0,T]$ and $k=1,2,...$, we get that for any $f\in L^2([0,T], \mathbb{R})$, 
$$
\| \boldsymbol{D}_t^{-1}f \|_{L^2} \leq \frac{1}{2 \eps} \|f\|_{L^2}, \quad \textrm{for all } 0\leq t \leq T. 
$$
Together with \eqref{op-norm} this completes the proof. 
\end{proof} 

\begin{lemma} \label{lem-bnd-gam} Let $\boldsymbol{\Gamma}^{-1}_t$ as in \eqref{op-gam-inv}. Then we have
$$
\sup_{t\leq T} \|\boldsymbol{\Gamma}^{-1}_t \|_{\rm{op}}   < \infty.
 $$
\end{lemma} 
\begin{proof} 
The proof follows directly from \eqref{op-gam-inv} and Lemma \ref{lem-d-bnd}, as each entry of $\Gamma^{-1}$ involves products of indicators and of $D_t^{-1}$.  
\end{proof} 
\begin{lemma} \label{lem-bnd-theta} Let $\Theta= \{\Theta_{t}(s) : t\in [0,s],\, s\in [0,T]\}$ as in \eqref{th-sol}. Then we have
 $$\E\left[  \sup_{t\leq T} \int_t^T \mathbb  |\Theta_t(s)|^2 ds\right]   <\infty. $$
 \end{lemma} 
\begin{proof} 
From \eqref{ind-def}, \eqref{th-sol}, Lemma \ref{lem-bnd-gam}, Fubini's Theorem and successive applications of Cauchy-Schwarz inequality we get 
$$
\begin{aligned} 
\sup_{t \leq T} |\Theta_t(s)|^2 &\leq  \sup_{t \leq T} \left( \| \boldsymbol {\Gamma}_t^{-1}\|_{\rm{op}}^2 \int_t^T \left(\E[P_{r} - P_T | \mathcal F_t ]\right)^2 dr \right) \\ 
& \leq C\int_0^T  \sup_{t \leq T}  \E[(P_{r} - P_T)^2| \mathcal F_t ]dr,  \quad \textrm{for all }  s \leq T,
\end{aligned} 
$$
where the constant $C>0$ is not depending on $s$.  Using Fubini's theorem it follows that, 
\be \label{th-sup-bnd} 
\begin{aligned} 
\E\left[\sup_{t \leq T} |\Theta_t(s)|^2 \right]  
& \leq C \int_0^T \E \left[ \sup_{t \leq T}  (\E[P_{r} - P_T| \mathcal F_t ])^2 \right]dr,  \quad \textrm{for all }   s \leq T,
\end{aligned} 
\ee
Together with \eqref{ass:P}, we conclude that 
$$
\E\left[  \sup_{t\leq T} \int_t^T \mathbb  |\Theta_t(s)|^2 ds\right]   <\infty,
$$
and we get the result. 
\end{proof} 

\begin{proof} [Proof of Lemma \ref{bnd-coef}] 
 (i) Recall that 
\be \label{bl0} 
a_t = \frac 1{2\lambda} \left(\E[ (P_t - P_T) \mid \mathcal F_t]  - \tilde h_0(t) + \langle \Theta_t , K_t \rangle_{L^2} +  \langle \boldsymbol {\Gamma}_t^{-1} K_t  , \mathds 1_t  ( \tilde h_0, q )^\top \rangle_{L^2} \right) 
\ee
{
From \eqref{ass:P} we get that
\be \label{bl1} 
   \E \left[\sup_{t\in [0,T]}  \left(\E[ (P_t - P_T) \mid \mathcal F_t] \right)^2 \right ] dt    <\infty. 
\ee
From \eqref{K-t}, \eqref{eq:assumtionG}, \eqref{g0-k-def}, Lemma \ref{lem-bnd-theta} and Cauchy-Schwarz inequality we have 
\be 
\begin{aligned} 
 \E \left[ \sup_{t\in [0,T]}\left(\langle \Theta_t , K_t \rangle_{L^2} \right)^2 \right ]   &=    \E \left[\sup_{t\in [0,T]} \left( \int_t^T \Theta_t(s) K(s,t) ds \right)^2 \right ]   \\ 
&\leq    \E \left[ \sup_{t\in [0,T]} \left( \int_t^T \Theta^2_t(s) ds  \int_t^TK^2(s,t) ds \right) \right ]  \\ 
& \leq    \E \left[\sup_{t\leq T} \left( \int_t^T \Theta^2_t(s) ds \right)\right]   \sup_{t\leq T}   \int_t^TK^2(s,t) ds \\
&< \infty.
\end{aligned} 
\ee
}
Note that from \eqref{K-t}, \eqref{eq:assumtionG}, \eqref{g0-k-def} we have $\sup_{t\leq T} \|K_t\|_{L^2} <\infty$. Together with Lemma \ref{lem-bnd-gam} we get
\bd
\begin{aligned} 
\sup_{t\leq T}  \int_{0}^T  \left( \int_{0}^T \boldsymbol {\Gamma}_t^{-1}(s,r)K_t(r)dr \right)^2ds &\leq
\sup_{t\leq T} \|\boldsymbol{\Gamma}^{-1}_t \|_{\rm{op}}^2 \sup_{t\leq T} \|K_t\|^2_{L^2}   \\ 
&< \infty. 
\end{aligned} 
\ed
Since by \eqref{z-def} $h_{0}$ and hence $\tilde h_{0}$ (by \eqref{h-tilde}) are square integrable deterministic functions, it follows yet again by Cauchy-Schwarz inequality that 
\be  \label{bl2} 
\sup_{t\leq T}  \langle \boldsymbol {\Gamma}_t^{-1} K_t  , \mathds 1_t  ( \tilde h_0, q )^\top \rangle_{L^2}  < \infty. 
\ee
Applying \eqref{bl1}--\eqref{bl2} into \eqref{bl0} gives (i). 

(ii) Recall that 
 \be \label{bl11} 
B(t,s) =  \mathds 1_{\{s<t\}}\frac 1{2\lambda} \left( \langle \boldsymbol {\Gamma}_t^{-1} K_t  , \mathds 1_t  ( \tilde G(\cdot,s), -1 )^\top \rangle_{L^2}   - \tilde G(t,s) \right). 
\ee 
Similarly to the derivation of \eqref{bl2} we have 
\be  \label{bl10} 
\sup_{t\leq T}  \langle \boldsymbol {\Gamma}_t^{-1} K_t  , \mathds 1_t  ( \tilde G(\cdot,s), -1 )^\top \rangle_{L^2}  < \infty, 
\ee
where we use \eqref{eq:assumtionG} and \eqref{g0-k-def} to bound $ \tilde G$. Then from \eqref{bl11} and \eqref{bl10} we get (ii). 

{(iii) For any $0<t \leq T$ we define 
$$
f(t) = \E\left[ \sup_{s\leq t} (u^*_s)^2 \right]. 
$$
From \eqref{volt-u} and Cauchy Schwarz inequality it follows that there exists positive constants $C_i(T)$, $i=1,2$ such that, 
\begin{align*}
f(t) 
 &\leq 2  \E\left[\int_0^ta^2_sds\right] +  2 \E\left[  \sup_{s\in [0,t]}\left(\int_0^s B(s,r)u_r^*dr \right)^2 ds \right] \\
 &\leq 2  \E\left[\int_0^Ta^2_sds\right] +  2 \E\left[ \sup_{s\in [0,t]} \left(\int_0^s B^2(s,r')dr' \int_0^s (u_{r}^*)^2dr \right) \right] \\
  &\leq 2  T\E\left[\sup_{s\in [0,T]} a^2_s\right] +  2 \left(  \sup_{s \leq T}\int_0^T B^2(s,r')dr' \right) \E\left[  \int_0^t \sup_{y \in [0,r] }  (u_{y}^*)^2dr \right] \\
 &\leq C_1(T) +   C_2(T)\int_0^t f(r) dr,  \quad \textrm{for all } 0\leq t \leq T, 
 \end{align*}
 where {we used} parts (i) and (ii) in the last inequality. Part (iii) then follows by  an application of Gr\"onwall inequality.
 }

\end{proof} 
\section{Proof of Lemma \ref{lem-id}} \label{sec-pf-id} 
 \begin{proof}[Proof of Lemma \ref{lem-id}]
 (i)  
Recall that $K_t$ was defined as a function $K_t:s\mapsto K(s,t)$. From \eqref{eq:diffkernelsigma} we note that for any $g\in L^2([0,T],\mathbb{R}^2)$ we have 
\be  \label{gg4} 
\begin{aligned}
 &-4\lambda  \langle \boldsymbol{\Psi}_t  \dot{\boldsymbol{\Sigma}}_t   f, g\rangle_{L^2}  \\
&=\int_0^T \int_0^T  \int_0^T  \Psi_t(s,r)  K(r,t)K(u,t)^\top f(u) g(s)du dr ds \\
&= \left(\int_0^T \int_0^T\Psi_t(s,r)  K(r,t)g(s)dr ds \right) \left(\int_{0}^{T} K(u,t)^\top f(u)du \right) \\
& =\langle \boldsymbol{\Psi}_t K_t, g \rangle_{L^2} \langle K_t,  f\rangle_{L^2}, 
\end{aligned}
\ee
and (i) follows by taking $f= \Theta_t$ and $g=g^u_t$.

(ii) Recall that $\boldsymbol{\Psi}_t$ is self-adjoint, then we have $  \Psi_t^\top(r,s) =   \Psi_t(s,r) $. It follows that 
\bn 
( \boldsymbol{\Psi}_tK_t(s) )^\top&=& \left(\int_0^T  \Psi_t(s,r)K_t(r)dr \right)^{\top}   \\ 
 &=&  \int_0^T K_t(r)^\top  \Psi_t(r,s) dr. 
\en

Using the fact that $\boldsymbol{\Psi}_t$ is self-adjoint and that $\mathds{1}_t(s)= \mathds{1}_{\{t\leq s\}}$, we get 
\be \label{fg1} 
\begin{aligned} 
  \langle \boldsymbol{\Psi}_tK_t, g^u_t   \rangle_{L^2} &= \int_0^T (\boldsymbol{\Psi}_tK_t)^{\top}(s)g^u_t(s) ds \\
 &= \int_0^T \int_0^T K_t(r)^\top  \Psi_t(r,s) g^u_t(r) dr ds \\
  &= \int_0^T \int_0^T K_t(r)^\top  \Psi_t(r,s) g^u_t(r)\mathds{1}_t(r) ds dr. 
\end{aligned} 
\ee
 On the other hand,
\be \label{fg2} 
\begin{aligned} 
  \left(\boldsymbol{\hat {K}}^*\boldsymbol{\Psi}_t \right)(g^u_t \mathds{1}_t)(t) =&\int_0^T (\boldsymbol{\hat {K}}^*\boldsymbol\Psi_t)(t,s) g^u_t(s)\mathds{1}_t(s) ds \\  
  =&\int_0^T \int_0^T {\hat {K}}^*(t,r) \Psi_t(r,s) g^u_t(s)\mathds{1}_t(s) ds dr. 
\end{aligned} 
\ee

Using \eqref{K-def}, we get for any $f\in L^2([0,T],\mathbb{R}^2)$, 
\be \label{fg3} 
\begin{aligned} 
\langle    K_t,f \rangle_{L^2}e_1 
&=   e_1 \int_0^T  K(r,t)^\top  f(r)dr  \\
&= -  \int_0^T   \hat K(r,t)^\top  f(r)dr  \\
&= - \left( \boldsymbol{\hat{K}}^*f\right)(t).
\end{aligned} 
\ee
Using \eqref{fg3} on \eqref{fg1} and \eqref{fg2} we get (ii).  

(iii) From \eqref{eq:diffkernelsigma} we get that, 
\begin{align*} 
  & -4\lambda \langle  \boldsymbol{\Psi}_t \dot{\boldsymbol{\Sigma}}_t \boldsymbol{\Psi}_t g_t^u, g_t^u \rangle_{L^2}   \\
   &=   \int_0^T    \int_0^T   \int_0^T \int_0^T (g_t^u(s))^\top  \Psi_t(s,v) K(v,t)   K(w,t)^\top \Psi_t(w,r)g_t^u(r)   dw dv dr ds  \\
   &= \int_0^T        (g^u_t(s))^\top \boldsymbol(\Psi_t K_t)(s)  ds  \int_0^T  \boldsymbol(\Psi_t K_t)^\top(r)  g_t^u(r)  dr \\ 
    &=   \langle \boldsymbol{\Psi}_t K_t  , g_t^u \rangle_{L^2}^2,
        \end{align*} 
where we used again the notation $K_t(s)= K(s,t)$. This completes the proof. 
\end{proof} 

\section{Proof of Proposition \ref{prop-theta} }\label{sec-bsde} 
\begin{proof} [Proof of Proposition \ref{prop-theta}]
 Let $0\leq s\leq T$.  An application of Lemma \ref{lemma-gam-phi}, yields that $\Theta$ given by  \eqref{th-sol} can be written as
 \be \label{th-sol2} 
\Theta_t(s)=- \left(\boldsymbol {\Psi}_t  \mathds{1}_t \E_t[P_{\cdot}- P_T] e_1 \right)(s) \quad \textrm{for all } 0\leq t\leq s.
\ee
 We first prove that $\Theta_\cdot(s)$ in \eqref{th-sol} satisfies \eqref{back-theta}. 
Recall the notation presented in Lemma~\ref{L:Psi}(i) and \eqref{psi-bar}. Together with \eqref{th-sol}, \eqref{ind-def} and \eqref{psi-gamma} we get that, 
\be \label{tte1} 
\begin{aligned}
\Theta_t(s) &=- \left(\boldsymbol {\Psi}_t \mathds{1}_t \E_t[P_{\cdot}- P_T] e_1 \right)(s) \\
&=  - \E_t[P_s- P_T]  A  e_1  - \int_t^T \bar {\psi}_t(s,r) \E_t[P_r- P_T]  e_1 dr, \quad \textrm{for all } 0\leq t \leq s.
\end{aligned}
\ee
From \eqref{psi-bar} it follows that $\dot{\boldsymbol{\Psi}}_t = \dot{\bar{\boldsymbol{\Psi}}}_t $.  We differentiate the above  expression for $\Theta_\cdot (s)$ with respect to $t$ and use \eqref{psi-bar} to get 
 \be \label{fgf} 
\begin{aligned}
&d\Theta_t(s)  \\
 &= - d\E_t[P_s- P_T]  A  e_1   - \left(\boldsymbol {\dot {\bar \Psi}}_t \mathds{1}_t \E_t[P_{\cdot}- P_T] e_1 \right)(s)   - \left(\bar{\boldsymbol  {\Psi}}_t \mathds{1}_t d\E_t[P_{\cdot}- P_T] e_1 \right)(s) \\
&\quad +  \bar{\psi}_t(s,t) \E_t[P_t-P_T]e_1\\
&=- d\E_t[P_s- P_T]  A  e_1   - \left(\boldsymbol {\dot {\Psi}}_t \mathds{1}_t \E_t[P_{\cdot}- P_T] e_1 \right)(s)  +  \bar{\psi}_t(s,t) \E_t[P_t-P_T]e_1 \\
&\quad - \left(\bar{\boldsymbol  {\Psi}}_t \mathds{1}_t d\E_t[P_{\cdot}- P_T] e_1 \right)(s) \\
&= 2 \left(\boldsymbol {  {\Psi}}_t  \dot{\boldsymbol{\Sigma}}_t \Theta_t  \right)(s)   +  \bar{\psi}_t(s,t) \E_t[P_t-P_T] e_1 +dN_t(s)  , 
\end{aligned}
\ee 
where we used \eqref{th-sol}, Lemma \ref{L:Psi}(iii) and 
\begin{align}\label{eq:dN}
dN_t(s) 
:=- \left({\boldsymbol  {\Psi}}_t \mathds{1}_t d\E_t[P_{\cdot}- P_T] e_1 \right)(s), 
\end{align}
\blue{From \eqref{th-sup-bnd} and  \eqref{ass:P} we have 
\be
 \begin{aligned} 
\E\left[\sup_{t \leq T} |\Theta_t(s)|^2 \right]  <\infty, \quad \textrm{for all }    s \leq T.
\end{aligned} 
\ee
Together with \eqref{tte1} and \eqref{ass:P} we get, 
\be \label{tte2}
\E\left[ \sup_{t \leq T} \left| \int_t^T \bar {\psi}_t(s,r) \E_t[P_r- P_T]  e_1 dr \right|^2 \right] <\infty.
\ee
 From Lemma \ref{L:Psi}(i) it follows that for all $t\in [0,T]$, $\boldsymbol{\Psi}_t = \bar{\boldsymbol{\Psi}}_t + A\cdot \id $ where $\bar {\psi}_t $ is the kernel of $\bar{\boldsymbol{\Psi}}_t$. Hence from \eqref{eq:dN} and \eqref{tte2} and application of \eqref{ass:P} to the operator $A\cdot \id$ it follows that,  
\begin{align}\label{eq:dN2}
\E \left[ \sup_{t\leq T}|N_t(s)|^2 \right]  <\infty, \quad \textrm{for all }   s \leq T.
\end{align}
Hence for any $s \in [0,T]$, $(N_t(s))_{t\in [0,s]}$ is a true martingale.  
}

Note that \eqref{back-theta} follows directly from \eqref{fgf}.
Next we prove that \eqref{th-sol2} satisfies the boundary condition \eqref{theta-bnd}. 
From  \eqref{th-sol} and Lemma \ref{L:Psi}(ii) we get that  
\be \label{jja} 
\begin{aligned} 
\Theta_s(s)&=- \left(\boldsymbol {\Psi}_s \mathds{1}_s \E_s[P_{\cdot}- P_T] e_1 \right)(s) \\
&= -\left( A\cdot\id + \frac{1}{2\lambda}\boldsymbol{\hat  {K}}^*\boldsymbol{\Psi}_s \right)( \E_s[P_{\cdot}- P_T] e_1  \mathds{1}_s)(s) \\
&=- Ae_1 \E_s[P_{s}- P_T] -\frac{1}{2\lambda} \left(  \boldsymbol{\hat  {K}}^*\boldsymbol{\Psi}_s \right)( \E_s[P_{\cdot}- P_T] e_1  \mathds{1}_s)(s) \\
&=- \frac{1}{2\lambda} e_1(P_s-M_s) -\frac{1}{2\lambda} \left(  \boldsymbol{\hat  {K}}^*\boldsymbol{\Psi}_s \right)( \E_s[P_{\cdot}- P_T] e_1  \mathds{1}_s)(s), 
\end{aligned} 
\ee
where we have used \eqref{a-mat} in the last equality. Next, we use \eqref{in-prod}, \eqref{th-sol2} and \eqref{eq:explicithatK} to get  
\bn 
(\boldsymbol{\hat  {K}}^*\boldsymbol{\Psi}_s )( \E_t[P_{\cdot}- P_T] e_1  \mathds{1}_s)(s) &=& {-} (\boldsymbol{\hat  {K}}^* \Theta_s(\cdot))(s) \\
&=& {-} \langle K_{s}, \Theta_s \rangle_{L^2} e_1, 
\en 
which together with \eqref{jja} {verifies \eqref{theta-bnd}}. 
 Finally, from Lemma \ref{lem-bnd-theta} we have 
 \begin{align}\label{eq:Thetasquareint}
    \sup_{t\leq T} \int_t^T \mathbb E\left[ |\Theta_t(s)|^2 \right] ds  <\infty, 
 \end{align}
 which completes the proof. 
\end{proof}

\section{Proof of Lemma~\ref{L:Psi}(ii)} \label{sec-pf-lem-psi} 
 Before we prove Lemma~\ref{L:Psi}(ii), we recall the notion of resolvent. For a kernel $K \in L^2([0,T]^2,\mathbb R^{2\times 2})$, we define its resolvent $R_T \in L^2([0,T]^2,\mathbb R^{2\times 2})$ by the unique solution to 
\begin{align}\label{eq:resolventeqkernel}
R_T = K + K \star R_T, \quad  \quad  K \star R_T =  R_T \star K.
\end{align} 
In terms of integral operators, this translates into 
\begin{align}\label{eq:resolventeqop0}
\boldsymbol{R}_T =  \boldsymbol{K} + \boldsymbol{K}\boldsymbol{R}_T, \quad \boldsymbol{K}\boldsymbol{R}_T=\boldsymbol{R}_T\boldsymbol{K}.
\end{align}
In particular, if $K$ admits a resolvent,  $({\id}-\boldsymbol{K})$ is invertible and
\begin{align}\label{eq:integralresop}
({\id}-\boldsymbol{K})^{-1}=\id+\boldsymbol{R}_T.
\end{align}
 
\begin{proof}[Proof of Lemma~\ref{L:Psi}(ii)]   
From Lemma A.2 of \cite{aj-22} we get the existence of the resolvent $\hat R$ of  $\frac {1} {2\lambda} \hat K$  which is again a Volterra kernel. From  \eqref{eq:integralresop} it follows that $(\id -\frac 1 {2\lambda} \boldsymbol{\hat{K}} )$ is invertible with an inverse given by $(\id + \boldsymbol{\hat{R}})$.  By Lemma \ref{lemmma-inv-op}, $\left( \id -  2\boldsymbol{\tilde \Sigma}_{t} A  \right)$ is invertible with an inverse given by $(\id +\boldsymbol{R}^{A}_t)$ where $\boldsymbol{R}^{A}_t$ is the resolvent of $ 2\boldsymbol{\tilde \Sigma}_{t} A$. We get that $\boldsymbol{\Psi}_t$ defined in \eqref{def:riccati_operator} satisfies 
\be
\begin{aligned}\label{eq:boundary}
	\boldsymbol{\Psi}_t &= (\id + \boldsymbol{\hat{R}})^* A (\id+\boldsymbol{R}^{A}_t)(\id+\boldsymbol{\hat{R}})   \\
	&= A \id  + \boldsymbol{\hat{R}}^* A  + A  \boldsymbol{\hat{R}}+ \boldsymbol{\hat{R}}^* A \boldsymbol{R}^{A}_t     + A \boldsymbol{R}^{A}_t    \boldsymbol{\hat{R}}    \\
	& \quad + \boldsymbol{\hat{R}}^* A  \boldsymbol{R}^{A}_t   \boldsymbol{\hat{R}}+ \boldsymbol{\hat{R}}^* A  \boldsymbol{\hat{R}}+ A  \boldsymbol{R}^{A}_t.
\end{aligned}
 \ee 
We first argue that 
\be\label{eq:bordR}
\hat{R}(s,u)=0, \quad  \mbox{for all }  0\leq s<u \leq T,
\ee
and
\be\label{eq:bordR2}
R^{A}_t(t, \cdot) = 0,  \quad \textrm{for all } 0\leq t\leq T.
\ee
Indeed, since $\hat K$ is a Volterra kernel, its resolvent $\hat R$ is also a Volterra kernel and \eqref{eq:bordR2} follows. From \eqref{eq:sigmakernel} we get that $\Sigma_t(t,\cdot)=0$ together with \eqref{def:C_tilde} we get that $\tilde \Sigma_t(t,\cdot)=0$, so that $R^{A}_t(t,\cdot)=0$  by the resolvent equation \eqref{eq:resolventeqkernel}. 

Let  $f \in L^2\left([0,T],\R^2\right)$. Using \eqref{eq:bordR} and \eqref{eq:bordR2} we get that,
\begin{align*}\label{eq:zero_terms}
	\left( A\boldsymbol{R}_t^{A}  \right) (f)(t) & =A  \int_0^T R^{A}_t(t,s)  f(s)ds \; = \; 0, \\
	\left(A \boldsymbol{R}^{A}_t  \boldsymbol{\hat{R}}\right)(f)(t) & = \;  A \int_0^T \int_0^T R^{A}_t(t,u)\hat{R}(u,s) f(s) du ds \; = \; 0.
\end{align*}
From \eqref{eq:resolventeqkernel} we get 
\be \label{res-id} 
\begin{aligned} 
\boldsymbol{\hat{R}}^* &= \frac {1} {2\lambda}\boldsymbol{\hat{K}}^* + \frac 1 {2\lambda}\boldsymbol{\hat{K}}^*\boldsymbol{\hat{R}}^* \\
&= \frac 1 {2\lambda}\boldsymbol{\hat{K}}^*(\id + \boldsymbol{\hat{R}}^*).
\end{aligned} 
\ee
Combining \eqref{res-id} with \eqref{eq:boundary} yields
\be \label{ps-r} 
\begin{aligned}
	(\boldsymbol{\Psi}_t) (f)(t) =& A (\id  +   \boldsymbol{\hat{R}})(f)(t)  +  (\boldsymbol{\hat{R}}^* A + \boldsymbol{\hat{R}}^*A \boldsymbol{R}^{A}_t+ \boldsymbol{\hat{R}}^*{A} \boldsymbol{R}_t^{A} \boldsymbol{\hat{R}}+ \boldsymbol{\hat{R}}^*{A} \boldsymbol{\hat{R}} )(f )(t)  \\
	=& A (\id  +   \boldsymbol{\hat{R}})(f)(t)  + \boldsymbol{\hat{R}}^* A (\id +  \boldsymbol{R}^{A}_t+  \boldsymbol{R}_t^{A} \boldsymbol{\hat{R}}+ 
	 \boldsymbol{\hat{R}} )f(t)  \\
	 =& A (\id  +   \boldsymbol{\hat{R}})(f)(t)  + \boldsymbol{\hat{R}}^* A (\id +  \boldsymbol{R}^{A}_t)(\id +  \boldsymbol{\hat{R}})
	  f(t)  \\
	   =& A (\id  +   \boldsymbol{\hat{R}})(f)(t)  + \frac{1}{{2\lambda}} \boldsymbol{\hat{K}}^*(\id+  \boldsymbol{\hat{R}}^*)A (\id +  \boldsymbol{R}^{A}_t)(\id +  \boldsymbol{\hat{R}})
	  f(t)  \\
	=&  A\left(\id  -\frac{1}{2\lambda} \boldsymbol{\hat {K}}  \right)^{-1}(f)(t) +\frac{1}{2\lambda} \left( \boldsymbol{\hat {K}}^*\boldsymbol{\Psi}_t \right)(f)(t).
\end{aligned}
\ee
Recall that $\mathds{1}_{t}(s) = \mathds{1}_{\{s \geq  t\}}$. From \eqref{eq:bordR} we have 
$$
 \boldsymbol{\hat{R}}(f\mathds{1}_t)(t) = \int_{0}^TR(t,u) f(u)\mathds{1}_{\{t\leq u\}} du= 0.  
$$
Together with \eqref{ps-r} we get the result of Lemma~\ref{L:Psi}(ii). 
\end{proof}

 \section{Proofs of Lemmas \ref{lemma-pos-def-ker}, \ref{lemmma-inv-op} and \ref{lemma-gam-phi}}\label{sec-lem-inv} 
 \begin{proof} [Proof of Lemma \ref{lemma-pos-def-ker}]
Assume that $G$ satisfies \eqref{g-spec} and let $f\in L^2\left([0,T],\mathbb R\right)$. Then using Fubini's theorem we get 
\bn
\int_{0}^T \int_0^T G(|t-s|) f(s)f(t) ds dt &=&  \int_{\mathbb{R_+} }\left(\int_{0}^T \int_0^Te^{-x|t-s|}f(s)f(t)ds dt \right) \mu(dx) \\
&\geq&0, 
\en
where we used the fact that for each $x\geq0$, $\wt G_x(t) = e^{-x t}$ is a nonnegative definite kernel (see Example 2.7 in \cite{GSS}) and that $\mu$ is a nonnegative measure. 
\end{proof}

 \begin{proof}[Proof of Lemma \ref{lemmma-inv-op}]
 We first note that from \eqref{eq:schur} it follows that $\boldsymbol{D}_t$ is a self-adjoint operator. We will show that under the assumptions of the lemma $\boldsymbol D_t$ is positive definite, hence it is invertible. 
 
 Recall that $\tilde G$ was defined in \eqref{h-tilde} and that $\boldsymbol{\tilde G}_t$ is the operator induced by the kernel $\tilde G(s,u)\mathds 1_{\{u \geq t\}}$. Clearly operator $\id$ is positive definite and $\boldsymbol{1}^*_t \boldsymbol{1}_t$ is nonnegative definite. It follows from \eqref{eq:schur} that in order to prove that $\boldsymbol D_t$ is positive definite we need to show that $ (\boldsymbol {\tilde G}_t +  \boldsymbol {\tilde G}^*_t)$ is nonnegative definite. 
Note that we can write the kernel of $\boldsymbol{\tilde G}_t$ as follows, 
\be \label{t-g-dec} 
\tilde G_t(s,u)= (2\varrho \mathds{1}_{\{u<s\}}+G(s,u)) \mathds{1}_{\{u > t\}}.   
\ee

Let $f\in L^2\left([0,T],\mathbb R\right)$, then from \eqref{pos-def} we get
\be \label{pd-1}
\begin{aligned}  
&\int_0^T\int_0^T\big( G_t(s,u)+G_t^*(s,u)\big) f(s)f(u)ds du  \\
&= \int_0^T\int_0^T\big( G(s,u)\mathds{1}_{\{u>t\}}+ G(u,s)\mathds{1}_{\{s>t\}} \big)f(s)f(u)ds du \\ 
&= \int_0^T\int_0^T\big( G(s,u)\mathds{1}_{\{s>u\}}\mathds{1}_{\{u>t\}}+ G(u,s)\mathds{1}_{\{s>t\}}\mathds{1}_{\{u>s\}} \big)f(s)f(u)ds du \\ 
&= \int_0^T\int_0^T\big( G(s,u)\mathds{1}_{\{s>u\}}\mathds{1}_{\{u>t\}}\mathds{1}_{\{s>t\}}+ G(u,s)\mathds{1}_{\{s>t\}}\mathds{1}_{\{u>s\}}\mathds{1}_{\{u>t\}} \big)f(s)f(u)ds du \\ 
&= \int_0^T\int_0^T\big( G(s,u) + G(u,s)  \big)f_t(s)f_t(u)ds du \\ 
&\geq 0, 
\end{aligned} 
\ee
where we used the fact that $G_t(s,u)=0$ for $u>s$, with $f_t(s):=f(s)\mathds{1}_{\{s>t\}}$. 

Moreover, we have
\be \label{pd-2} 
\begin{aligned}
&  \int_{t}^T\int_{u}^T f(s)f(u)dsdu  + \int_{0}^T\int_{0}^T \mathds 1_{\{t\leq s \leq u\}}f(u)f(s)du ds \\
&=   \int_{t}^T\int_{u}^T f(s)f(u)dsdu  + \int_{t}^T\int_{t}^u f(u)f(s)ds du \\
&=   \int_{t}^T\int_{t}^T f(s)f(u)dsdu \\
&= \left( \int_{t}^T f(s)ds\right)^2  \\
&\geq 0.
\end{aligned}
\ee
From \eqref{t-g-dec}, \eqref{pd-1} and \eqref{pd-2} it follows that $\boldsymbol{\tilde G}_t$ is nonnegative definite and this completes the proof.  
\end{proof} 

We now turn to the proof of  Lemma \ref{lemma-gam-phi}.  First, we need an auxiliary result.
Recall that $K$ was defined in \eqref{g0-k-def}.
We define  $ \boldsymbol{K}_t$ as the operator induced by the kernel $K(s,u)\mathds 1_{\{u \geq t\}}$  and
\be \label{k-hat-e} 
\boldsymbol{\hat K}_t = - ( \boldsymbol{K}_t \otimes e_1),
\ee
which by \eqref{eq:explicithatK} is induced by the kernel 
\be \label{k-hat-mat} 
\hat K_t(s,u) =\begin{pmatrix}
-\tilde G(s,u)\mathds 1_{\{u \geq t\}} &  0 \\
\mathds{1}_{\{s\geq u\}}\mathds 1_{\{u \geq t\}}  & 0  \\
 \end{pmatrix}.
 \ee
 Recall that $ \boldsymbol{\hat R}$ was defined after \eqref{eq:integralresop} and that $\boldsymbol{ \hat K}_t$ was defined in \eqref{k-hat-e}.
\begin{lemma}
      Let $\boldsymbol {\hat R}$ be the resolvent of $\boldsymbol {\hat K}$ and let $\boldsymbol {\hat R}_t$ be the operator induced by the kernel $\hat{R} (s,u)\mathds 1_{\{u \geq t\}}$. Then $  \boldsymbol {\hat R}_t$ is the resolvent of $\boldsymbol{ \hat K}_t$.  
\end{lemma}
  \begin{proof}
Recall that from \eqref{eq:resolventeqkernel} we have 
\be  \label{jj2} 
 \hat R(s,u) =\hat K(s,u) + \int_0^T \hat R(s,z) \hat K(z,u)dz,
\ee
and 
\be \label{jj3} 
  \int_0^T \hat K(s,z)\hat R(z,u)dz = \int_0^T \hat R(s,z)\hat K(z,u) dz.
\ee
Using \eqref{jj2} we can write $\hat R_t$ as follows 
\be\label{jj0} 
\begin{aligned}
  \hat R_t(s,u)&= \hat R(s,u)\mathds 1_{\{u \geq t\}}\\ 
   &= \hat K(s,u)\mathds 1_{\{u \geq t\}} + \int_0^T \hat R(s,z)\hat K(z,u)\mathds 1_{\{u \geq t\}} dz\\
   &= \hat K_t(s,u) + \int_0^T \hat R_t(s,z)\hat K_t(z,u) dz.
\end{aligned}
\ee
Since $\hat R$ is a Volterra kernel and by \eqref{jj3} we get, 
\be \label{jj1} 
\begin{aligned}
    \int_0^T \hat K_t(s,z)\hat R_t(z,u)dz &= \int_0^T \hat K(s,z)\mathds 1_{\{z \geq t\}}\hat R(z,u)\mathds 1_{\{u \geq t\}}dz \\
    &= \mathds 1_{\{u \geq t\}}  \int_0^T \hat K(s,z) \hat R(z,u)dz\\ 
    &=\mathds 1_{\{u \geq t\}} \int_0^T \hat R(s,z) \hat K(z,u)dz  \\
    &= \int_0^T \hat R_t(s,z) \hat K_t(z,u)dz. 
\end{aligned}
\ee
From \eqref{jj0} and \eqref{jj1} it follows that $ {\hat R}_t$ and ${ \hat K}_t$ satisfy \eqref{eq:resolventeqkernel}, and the result follows. 
  \end{proof}

\begin{proof} [Proof of Lemma \ref{lemma-gam-phi}] 
We first prove that the operator $\boldsymbol \Psi_{t}$ in \eqref{def:riccati_operator} is well defined. From \eqref{eq:schur}, \eqref{op-gamma} and Lemma \ref{lemmma-inv-op} it follows that $ \boldsymbol \Gamma_t$ is well defined and invertible.  We will use the invertibility of $\boldsymbol{\Gamma}_t$ to show that $\left(\id  - 2   \tilde{\boldsymbol{\Sigma}}_t A    \right) $ is invertible and compute its inverse. Since the invertibility  of $\left( \id - \frac 1 {2\lambda} \boldsymbol{\hat K}^* \right) $ and $\left( \id - \frac 1 {2\lambda} \boldsymbol{\hat K} \right)$ is given as a resolvent of a Volterra operator (cf. \eqref{eq:integralresop}), this will prove that $\boldsymbol \Psi_{t}$ is well defined. 
  
We start by deriving an essential identity. 
Since $\Sigma_t(s,u)=0$ if $s\vee u \leq t$ we can rewrite \eqref{def:C_tilde} as follows:
 \begin{align}\label{eq:tildesigmat}
\boldsymbol{\tilde{\Sigma}}_t &=   \left( \id - \frac 1 {2\lambda} \boldsymbol{\hat K} \right)^{-1} \boldsymbol{\Sigma}_t \left(\id - \frac 1 {2\lambda} \boldsymbol{\hat K}^*\right)^{-1}\\
&=\left( \id +  \boldsymbol{\hat R}\right) \boldsymbol{\Sigma}_t \left(\id +\boldsymbol{\hat R}^*\right)\\
&=  \boldsymbol{\Sigma}_t  + \boldsymbol{\Sigma}_t \boldsymbol{\hat R}^* + \boldsymbol{\hat R}\boldsymbol{\Sigma}_t  + 
 \boldsymbol{\hat R}\boldsymbol{\Sigma}_t \boldsymbol{\hat R}^*\\
 &=  \boldsymbol{\Sigma}_t  + \boldsymbol{\Sigma}_t \boldsymbol{\hat R}^* + \boldsymbol{\hat R}\boldsymbol{\Sigma}_t  + 
 \boldsymbol{\hat R}_t\boldsymbol{\Sigma}_t \boldsymbol{\hat R}^*_t\\
 &= \left( \id +  \boldsymbol{\hat R}_t\right) \boldsymbol{\Sigma}_t \left(\id +\boldsymbol{\hat R}^*_t\right)\\
 &= \left( \id - \frac 1 {2\lambda} \boldsymbol{\hat K}_t \right)^{-1} \boldsymbol{\Sigma}_t \left(\id - \frac 1 {2\lambda} \boldsymbol{\hat K}_t^*\right)^{-1}.
\end{align}

Using \eqref{gam-0} we can write 
\begin{align}
&    {\boldsymbol \Gamma}_t \\
&= A^{-1} \id  - \frac 1 {2\lambda} \left( \boldsymbol{\hat K}_tA^{-1}  + A^{-1}  \boldsymbol{\hat K}_t^*   \right) \\
    &=  A^{-1} \id  - \frac 1 {2\lambda} \left( \boldsymbol{\hat K}_tA^{-1}  + A^{-1}  \boldsymbol{\hat K}_t^*   \right)  + \frac{1}{4\lambda^2}  \boldsymbol{\hat K}_t  A^{-1}  \boldsymbol{\hat K}_t^* - 2  \boldsymbol{\Sigma}_t \\
    &= \left( \id - \frac 1 {2\lambda} \boldsymbol{\hat K}_t\right) A^{-1}  \left( \id - \frac 1 {2\lambda} \boldsymbol{\hat K}_t^* \right) - 2 \boldsymbol{\Sigma}_t   \\
    &=    \left( \id - \frac 1 {2\lambda} \boldsymbol{\hat K}_t\right) \left(\id  - 2  \left( \id - \frac 1 {2\lambda} \boldsymbol{\hat K}_t\right)^{-1} \boldsymbol{\Sigma}_t  \left( \id - \frac 1 {2\lambda} \boldsymbol{\hat K}^*_t\right)^{-1} A  \right)   A^{-1}  \left( \id - \frac 1 {2\lambda} \boldsymbol{\hat K}_t^* \right)\\
    &=   \left( \id - \frac 1 {2\lambda} \boldsymbol{\hat K}_t\right) \left(\id  - 2   \tilde{\boldsymbol{\Sigma}}_t A    \right)   A^{-1}  \left( \id - \frac 1 {2\lambda} \boldsymbol{\hat K}_t^* \right),
\end{align}
where in the last equality we used \eqref{eq:tildesigmat}.  Now since $\left( \id - \frac 1 {2\lambda} \boldsymbol{\hat K}_t\right)$ and $\left( \id - \frac 1 {2\lambda} \boldsymbol{\hat K}_t^* \right)$ are invertible, it follows that, 
\begin{align}
    \left(\id  - 2   \tilde{\boldsymbol{\Sigma}}_t A    \right)  
      &=  \left( \id - \frac 1 {2\lambda} \boldsymbol{\hat K}_t\right)^{-1}  \boldsymbol{\Gamma}_t  \left( \id - \frac 1 {2\lambda} \boldsymbol{\hat K}^*_t\right)^{-1} A, 
\end{align}
which proves that $   \left(\id  - 2   \tilde{\boldsymbol{\Sigma}}_t A    \right)$ is invertible with an inverse which is given by
\begin{align}
    \left(\id  - 2   \tilde{\boldsymbol{\Sigma}}_t A    \right)^{-1}  = A^{-1} \left( \id - \frac 1 {2\lambda} \boldsymbol{\hat K}^*_t\right)     \boldsymbol{\Gamma}_t^{-1}  \left( \id - \frac 1 {2\lambda} \boldsymbol{\hat K}_t\right).
    \end{align}

Next, we prove relation \eqref{psi-gamma}. We first note that  both terms appearing in  expression \eqref{psi-gamma} are continuous in  the parameter $\phi \in [0,\infty)$, recall \eqref{op-gam-inv} and \eqref{def:riccati_operator}. It is therefore enough to prove \eqref{psi-gamma} for any $\phi>0$.
Note that for $A$ in \eqref{a-mat} we have for $\phi>0$,  
\be \label{a-mat-inv} 
A^{-1} = \left(\begin{matrix}
 {2\lambda} & 0 \\
0 & - \frac 1{2\phi}
\end{matrix}\right).
\ee
From \eqref{eq:schur}, \eqref{op-gamma}, \eqref{a-mat-inv} and  \eqref{op-one} we obtain
\be \label{gam-0}
\begin{aligned}
\boldsymbol{\Gamma}_t 
&= \left(\begin{matrix}
 {2\lambda} \id + (\boldsymbol{\tilde G}_t + \boldsymbol{\tilde G}^*_t) &- \boldsymbol{1}^*_t \\
-\boldsymbol{1}_t & - \frac 1{2\phi} \id 
\end{matrix}\right)\\
&= A^{-1} \id  - \frac 1 {2\lambda} \left( \boldsymbol{\hat K}_tA^{-1}  + A^{-1}  \boldsymbol{\hat K}_t^*   \right).
\end{aligned}
\ee
Using Lemma \ref{lemmma-inv-op} and \eqref{def:riccati_operator} we note that in order to prove Lemma \ref{lemma-gam-phi}, it is enough to prove that for any $f\in L^2\left([0,T],\mathbb R\right)$, the quantities
\be \label{to-s} 
\boldsymbol{\Psi}^{-1}_t f\mathds 1_t   =   \left( \id - \frac 1 {2\lambda} \boldsymbol{\hat K} \right)  \left( \id -  2\boldsymbol{\tilde \Sigma}_{t} A  \right)  A^{-1}  \left( \id - \frac 1 {2\lambda}  \boldsymbol{\hat K^*} \right) f\mathds 1_t,  \quad  t\leq T, 
\ee
coincide with the left-hand side of \eqref{gam-0} operating on $f\mathds 1_t$. 


Let $f\in L^2\left([0,T],\mathbb R\right)$. From \eqref{def:C_tilde} and \eqref{to-s} we get,   
\be \label{to-s1} 
 \boldsymbol{\Psi}^{-1}_t  f\mathds 1_t  =  \left( \left( \id - \frac 1 {2\lambda} \boldsymbol{\hat K}\right) A^{-1}  \left( \id - \frac 1 {2\lambda} \boldsymbol{\hat K}^* \right) - 2 \boldsymbol{\Sigma}_t  \right)  f\mathds 1_t.
\ee
From \eqref{k-hat-mat} we get 
\begin{equation}\label{cor1} 
\left(\boldsymbol {\hat K} f\mathds 1_t\right)(s) = \int_0^T \hat K(s,u) \mathds 1_t(u) f(u) du = \int_0^T \hat K_t(s,u) f(u) du = \left(\boldsymbol {\hat K}_t f\mathds 1_t\right)(s),
\end{equation} 
and 
\be\label{cor2}
K_t(s,u)^* = K_t(u,s)^\top = K(u,s)^\top  \mathds 1_t(s). 
\ee
From \eqref{cor2} we get
\be \label{cor3} 
\begin{aligned}
\mathds{1}_t(s)\left(\boldsymbol {\hat K}^* f\mathds 1_t\right)(s) &= \int_0^T \mathds{1}_t(s)\hat K^*(s,u) \mathds 1_t(u) f(u)du\\ 
&= \int_0^T \mathds{1}_t(s)\hat K(u,s)^\top \mathds 1_t(u) f(u) du \\
&=\mathds{1}_t(s)\left(\boldsymbol {\hat K}^*_t f\mathds 1_t\right)(s).
\end{aligned}
\ee
From \eqref{to-s1}, \eqref{cor1} and \eqref{cor3} it follows that 
\be \label{to-s2} 
\begin{aligned}
&\mathds{1}_t(s)\big(\boldsymbol{\Psi}^{-1}_t  f\mathds 1_t \big)(s)  \\
&=\mathds{1}_t(s)\left( \Big( A^{-1} \id  - \frac 1 {2\lambda} \left( \boldsymbol{\hat K}_tA^{-1}  + A^{-1}  \boldsymbol{\hat K}_t^*   \right) + \frac{1}{4\lambda^2}  \boldsymbol{\hat K}  A^{-1}  \boldsymbol{\hat K}^* - 2  \boldsymbol{\Sigma}_t \Big)f\mathds 1_t\right)(s). 
\end{aligned}
\ee
To see that we recall that for $K$ and $\hat K$ were defined in \eqref{g0-k-def} and \eqref{eq:explicithatK}, respectively. A direct matrix multiplication, using \eqref{a-mat-inv} gives
$$
\boldsymbol{\hat K}  A^{-1}  \boldsymbol{\hat K} = 2 \lambda  \boldsymbol{K} \boldsymbol{K}^*. 
$$
Together with \eqref{eq:sigmakernel} we get, 
\begin{align*}
 \mathds{1}_t(s)  \big(\boldsymbol{\hat K}  A^{-1}  \boldsymbol{\hat K}^*f\mathds 1_t \big)(s)    = 2 \lambda  \mathds{1}_t(s)  \big( \boldsymbol{K} \boldsymbol{K}^*{f}\mathds 1_t \big)(s) =  {8\lambda^2} \mathds{1}_t(s) \big( \boldsymbol {\Sigma_t}f\mathds 1_t \big)(s). 
\end{align*}
 
Hence the last two terms in the right-hand side of \eqref{to-s2} cancel, and together with \eqref{gam-0} we get 
\be \label{cor-4} 
\begin{aligned} 
\mathds{1}_t(s) \big( \boldsymbol{\Psi}^{-1}_t  f\mathds 1_t \big)(s) =&\mathds{1}_t(s)\left( \Big( A^{-1} \id  - \frac 1 {2\lambda} \left( \boldsymbol{\hat K}_tA^{-1}  + A^{-1}  \boldsymbol{\hat K}_t^*   \right) \Big) f\mathds 1_t \right)(s) \\ 
=&\mathds{1}_t(s) \big(\boldsymbol{\Gamma}_t f\mathds 1_t \big)(s), 
\end{aligned} 
\ee
and the result follows.
\end{proof} 

 Conflict of Interest Statement \\
All authors declare no conflicts of interest. \\

 Data Availability Statement \\
Data sharing is not applicable to this article as no new data were created or analyzed in this study.


\end{document}